\documentclass[11pt]{article}
\usepackage[]{geometry}
\usepackage{graphicx}
\usepackage{caption}
\usepackage{subcaption}
\usepackage{hyperref}
\usepackage{amsmath}
\usepackage{amsthm}
\usepackage{amssymb}
\usepackage{authblk}

\usepackage{bm}
\usepackage{xcolor}
\usepackage[dvipsnames]{xcolor}
\usepackage{float}
\usepackage{booktabs}

\usepackage{yhmath}
\usepackage{mathdots}
\usepackage{MnSymbol}

\usepackage{caption}
\usepackage[utf8]{inputenc}
\usepackage{tikz}
\usepackage{pgfplots}
\pgfplotsset{compat=1.17}
\usepackage{authblk}
\usepackage{pgfplots}
\usepgfplotslibrary{fillbetween}
\usepackage{bm}

\everymath{\displaystyle}
\allowdisplaybreaks

\newtheorem{theorem}{Theorem}[section]

\newtheorem{proposition}{Proposition}[section]
\newtheorem{remark}{Remark}[section]
\newtheorem{corollary}{Corollary}[section]

\DeclareMathOperator*{\minmod}{minmod}

\DeclareMathOperator*{\sech}{sech}

\newcommand{\dd}{\mathrm{d}}

\providecommand{\keywords}[1]
{ \small \textbf{\textit{Keywords---}}#1 }

\title{  A Direct Approach for Detection of Bottom Topography in Shallow Water}
\author[$\;$]{}

\affil[$\;$]{ {\large Noureddine Lamsahel}}
\affil[$\;$]{ {\scriptsize  University of Littoral C\^ote d'Opale, Laboratory of Pure and Applied Mathematics, 50 Rue F. Buisson, 62228 Calais-Cedex, France.}}
\affil[$\;$]{ {\scriptsize  Mohammed VI Polytechnic University, The UM6P Vanguard Center, Benguerir 43150, Lot 660, Hay Moulay Rachid, Morocco.}}
\affil[$\;$]{ {\scriptsize Emails: noureddine.lamsahel@um6p.ma, noureddine.lamsahel@etu.univ-littoral.fr  }}
\affil[$\;$]{ {\large }}
\affil[$\;$]{ {Carole Rosier}}
\affil[$\;$]{ {\scriptsize  University of Littoral C\^ote d'Opale, Laboratory of Pure and Applied Mathematics, 50 Rue F. Buisson, 62228 Calais-Cedex, France.}}
\affil[$\;$]{ {\scriptsize Email: carole.rosier@univ-littoral.fr}}
\affil[$\;$]{ {\large }}

\date{}

\begin{document}

\maketitle

\begin{abstract}
We propose a fast, stable, and direct analytic method to detect underwater channel topography from surface wave measurements, based on one-dimensional shallow water equations. The technique requires knowledge of the free surface and its first two time derivatives at a single instant $t^{\star}$ above the fixed, bounded open segment of the domain. We first restructure the forward shallow water equations to obtain an inverse model in which the bottom profile is the only unknown, and then discretize this model using a second-order finite-difference scheme to infer the floor topography.  We demonstrate that the approach satisfies a Lipschitz stability and is independent of the initial conditions of the forward problem. The well-posedness of this inverse model requires that, at the chosen measurement time $t^{\star}$, the discharge be strictly positive across the fixed portion of the open channel, which is automatically satisfied for steady and supercritical flows.  For unsteady subcritical and transcritical flows, we derive two empirically validated sufficient conditions ensuring strict positivity after a sufficiently large time. The proposed methodology is tested on a range of scenarios, including classical benchmarks and different types of inlet discharges and bathymetries. We find that this analytic approach yields high approximation accuracy and that the bed profile reconstruction is stable under noise. In addition, the sufficient conditions are met across all tests.
\end{abstract}

\noindent\keywords{Bathymetry, Inverse problem, Shallow water equations, Free surface measurements, Direct inverse approach, Analytic reconstruction.}

\section{Introduction}
Accurate knowledge of underwater geometry in rivers and coastal zones is critical for coastal engineering applications, including the study of open channel flow hydrodynamics \cite{marks2000integration,cunge1980practical} and the prediction of tsunami run-up in diverse settings \cite{synolakis1987runup,synolakis1991green,grilli1994shoaling}. This prediction direction is referred to as the direct problem and, importantly, utilizes bathymetry information to explore its impact on surface wave deformation, as well as changes in velocity and discharge across the domain. The selection of an appropriate forward model depends on the physical context and flow characteristics. For inviscid flows, we can employ the full incompressible Euler system or adopt its asymptotic simplification, including the Saint‑Venant (shallow water) equations \cite{cunge1980practical,synolakis2008validation}.

Detecting the bottom profile with good accuracy is a challenging problem that often requires direct measurement techniques, which are time-consuming and costly. One of the primary approaches uses devices based on acoustic waves \cite{smith2004conventional}. However, this technique is slow, incurs high operational costs, and has a risk of stranding \cite{sagawa2019satellite,carron2001proposed}. As a second technique, in \cite{irish1999scanning,hilldale2008assessing} the authors propose collecting bed topography measurements using airborne LiDAR. Although faster than sonar surveys, this approach remains costly and is limited by water clarity and flight restrictions \cite{hilldale2008assessing,sagawa2019satellite}. A comprehensive overview of additional direct bathymetric survey methods and their limitations can be found in the review paper \cite{lecours2016review}  and the references therein.

As an alternative to these direct measurement techniques, we can infer the bottom profile using free surface measurements, which are significantly more accessible than direct underwater data \cite{hilldale2008assessing,smart2009river}. This strategy is formulated as an inverse problem based on a suitably chosen governing equations. For different models used to detect the bed profile from free surface data, see the review paper \cite{sellier2016inverse}. In this work, we formulate the recovery of bathymetry as an inverse problem for the one‑dimensional shallow water equations. Importantly, our approach does not rely on the zero‑inertia approximation \cite{gessese2013bathymetry} nor the small wave amplitude assumption \cite{cocquet2021optimization}. In this direction, there are mainly two approaches used to solve this inverse problem: direct numerical methods (one‑shot) \cite{gessese2011reconstruction,gessese2011inferring} and iterative schemes based on PDE‑constrained optimization \cite{angel2024bathymetry,khan2021variational,khan2022variational}. In the one‑shot reconstruction paradigm, one first discretizes the forward shallow water equations and then manipulates the discrete system by swapping the roles of the free surface elevation and the bottom profile. This scheme requires access to the free surface elevation, the inlet discharge, and the bed elevation at the upstream boundary of the open channel. However, this method assumes steady‑state flow. Moreover, there is no rigorous analysis of its convergence, stability, or independence from the initial conditions of the forward model. The only exception is \cite{gessese2011inferring}, where a pseudo‑analytic solution of the bed profile is derived in the steady one‑dimensional case. On the other hand, the optimization approaches are free from the steady‑state assumption; however, they require the same measurements as the direct methods, collected over a sufficiently long time interval. In particular, in \cite{khan2021variational,khan2022variational}, a limited number of observation locations of the free surface is considered, yet the numerical experiments show that increasing the number of observation points improves convergence. In addition to the measurements used by the direct approach, the optimization-based approaches implicitly necessitate specified initial velocity and outlet depth to numerically solve the shallow water system and its adjoint. Moreover, a smallness condition on the amplitudes of the free surface and the bathymetry is required for the convergence in \cite{khan2021variational,khan2022variational}.

In this work, we follow the one‑shot approach introduced in \cite{gessese2011reconstruction}. Importantly, we refrain from assuming steady‑state flow. Moreover, instead of manipulating a discrete shallow-water scheme, we first derive an inverse system from the continuous forward model, where the bottom profile is the only unknown. We then discretize this inverse system using a second‑order finite‑difference method.  Crucially, our formulation is inherently independent of the initial conditions of the forward model and relies on minimal measurements that naturally extend those of the steady‑state method \cite{gessese2011reconstruction}. Furthermore, we prove a Lipschitz stability estimate for our approach, which is an important guarantee for this class of ill‑posed inverse problems. In the end, we confirm the accuracy and robustness of the approach  through various numerical tests. The Python code is available in \cite{codepyhon}, allowing the reproduction of all numerical results of this paper and the development of further test cases.

\subsection{Main results}
We denote the elevation of the free surface by $\zeta$. Assume that, at a single time $t^{\star}$, we have access to the following three free surface measurements:
\begin{equation}\label{surafcemeasurement}
S_m(t^{\star})=\left(\zeta,\,\partial_t\zeta,\,\partial_t^2\zeta\right)|_{\{t^{\star}\}\times I}, 
\end{equation}
 throughout a fixed channel interval $I=(a_1,a_2)$. Equivalently, we may replace the single instant measurements \eqref{surafcemeasurement} by knowing $\zeta$ on a small time interval $(t^{\star}-\epsilon,t^{\star}+\epsilon)\times I$. Note that these free surface measurements are less restrictive than assuming access to a particular free surface profile over an interval of time across the domain, such as standing periodic waves in \cite{nicholls2009detection}. By supplementing \eqref{surafcemeasurement} with the inlet discharge information $\left(q(a_1,t^{\star}),\partial_t q(a_1,t^{\star})\right)$ and the bed elevation at the upstream boundary $b(a_1)$, see Figure \ref{fig:inverse_domain}, our main results can then be summarized as follows:
\begin{enumerate}
\item Under the nondegeneracy condition of the discharge \eqref{condition}, we derive from the forward shallow water system \eqref{sh2} a direct inverse model \eqref{sh4} for the bed profile detection. Note that this condition on the discharge holds automatically in supercritical and steady flows. In particular, for the steady‑state case, we obtain a fully analytic expression for the bottom profile \eqref{sh5} by integrating the inverse model \eqref{sh4}.  Crucially, the required inlet discharge $q(a_1)=constant$ can be replaced by access to the bed elevation at a specific location $x_{\ast}$ \eqref{simpliciation for discharge}. Moreover,  the measurements \eqref{surafcemeasurement} reduce to the free surface elevation alone.
\item We establish two sufficient conditions \eqref{condition on bottom} and \eqref{condition on bottom modified} that guarantee the nondegeneracy assumption of the discharge for any sufficiently large time in subcritical and transcritical cases. These criteria relate inlet/outlet velocities and depths, plus a term involving the bed profile variation, which is controlled by the right‑ and left‑going characteristic speeds. We then prove a Lipschitz stability estimate \eqref{stability} for the inverse model \eqref{sh4}.
\item Implementing the inverse model \eqref{sh4} via a second‑order finite‑difference method, we perform several numerical tests that include classical benchmarks to demonstrate both the high accuracy and robustness of the method. Moreover, in all these numerical tests, and depending on the flow, one or both sufficient conditions \eqref{condition on bottom} and \eqref{condition on bottom modified} hold. These experiments also illustrate the influence of initial conditions, boundary data, and bed geometry on flow characteristics.
\end{enumerate}

\subsection{Outline of the paper}
In Sect. \ref{sec:inverse direct}, we briefly introduce the one‑dimensional shallow water equations and develop the corresponding inverse model for recovering the channel bed using the three free surface measurements taken at a single time \eqref{surafcemeasurement}. We then impose a strict positivity condition on the discharge across the open channel to guarantee the feasibility of a one‑shot method. The section concludes with the steady‑flow scenario, where we solve the inverse model to derive an explicit analytic reconstruction formula for the bed topography. In Sect. \ref{sec:guaranteeandstability}, we establish two (practically feasible) sufficient conditions under which the flow rate remains strictly positive in subcritical and transcritical flows. These two conditions provide a guarantee of nondegeneracy for the inverse model after a sufficiently large time. Both conditions depend on the ability of the characteristic speeds to dominate bottom variation. We end this section by proving a Lipschitz stability estimate for our reconstruction approach. In Sect. \ref{sec:numericalexpr}  we conduct several numerical tests to qualify the bathymetry reconstruction accuracy and the method stability.  Each test verifies the two sufficient conditions to confirm that the discharge remains strictly positive after a sufficiently large time. Finally, Sect. \ref{sec:conclusions} summarizes our findings, provides some further comments, and outlines possible directions for future work.

\section{The Forward and the inverse problems}\label{sec:inverse direct}
This section begins with a brief overview of the one‑dimensional shallow water equations and their principal features. Moreover, we present two mathematically equivalent formulations of this system and highlight an important physical distinction that we will exploit in our bed profile detection. In Subsection \ref{inverseproblemintro}, we algebraically manipulate the forward system to derive an inverse model that only needs  \eqref{surafcemeasurement} to infer the bathymetry. We conclude Subsection \ref{inverseproblemintro} by deriving an explicit analytic expression for the bed profile in the steady-flow case, where the usual inlet discharge requirement can be substituted by knowing the bed elevation at a second specific spatial point.
\subsection{Governing equations }
We consider the motion of a layer of an ideal, inviscid, incompressible, and irrotational fluid in an open channel. The fluid domain is bounded below by a solid bottom and above by a free surface, with both boundaries assumed to be parametrized, see Figure \ref{fig:inverse_domain}. Assuming the shallow water regime, $\sigma = \frac{h}{\lambda} \ll 1,$ where $h$ is the depth of the flow and $\lambda$ is the wavelength. Averaging the two-dimensional Euler equations over the vertical axis yields the non-dispersive \(\mathcal{O}(\sigma^2)\) one-dimensional shallow water equations \cite{de1871theorie,lannes2013water}:
\begin{equation}\label{sh1}
\begin{cases}
\displaystyle \partial_{t}\zeta + \partial_{x}\bigl((\zeta-b)\,u\bigr) = 0,\\[6pt]
\displaystyle \partial_{t}u + u\,\partial_{x}u + g\,\partial_{x}\zeta = 0,
\end{cases}
\end{equation}
where $u$ is the depth-averaged horizontal velocity, $\zeta=h+b$ is the free surface elevation, $b$ is the bed topography (bathymetry) to be reconstructed, and $g$ is the gravitational acceleration set to $g = 9.812 \;m/s^2$.

In system \eqref{sh1}, the friction resistance of the bathymetry $b$ is neglected under the inviscid flow assumption. Nonetheless, numerical investigations have shown that the friction term exerts a negligible influence on both the forward and inverse problems, see \cite{forbes1988critical,fadda1997open,tam2015predicting,angel2024bathymetry}. Therefore, the Euler-based model \eqref{sh1} is practically valid for the study of the impact of bathymetry on surface wave deformation and for reconstructing bed topography from a given admissible measurements, under the assumption of the shallow-water regime. With straightforward algebraic manipulations,  one can transform the above system \eqref{sh1} posed in terms  of velocity and free surface elevation, to the following  equivalent (discharge, depth) format:
\begin{equation}\label{sh2}
\begin{cases}
\displaystyle \partial_{t}h + \partial_{x}q = 0,\\[6pt]
\displaystyle \partial_{t}q + \partial_{x}\left(\frac{q^2}{h}\right) + g\,h\,\partial_{x}(h+b)= 0,
\end{cases}
\end{equation}
where $q=h\,u$ is the discharge (flow rate). Although the systems \eqref{sh1} and \eqref{sh2} are mathematically equivalent, they are physically different. For the moment, in practical applications, the discharge $q$ in the domain can be measured by direct measurement \cite{huang2018detecting}. As an alternative approach, we can use the first equation of the above system \eqref{sh2} ($\partial_x q=-\partial_t\zeta$) to infer the discharge in the domain only by knowing the first time derivative of free surface elevation ($\partial_t\zeta$), along with the inlet discharge at the upstream of the channel. This provides an easier method for approximating discharge using free surface measurements. In contrast, determining the depth-averaged horizontal velocity $u$ is more involved and may require a known relation that connects $u$ to the fluid’s surface velocity (see, for example, \cite{lee2006electromagnetic}), which, in turn, should be approximated using direct measurement techniques \cite{adrian2011particle}. In this study, we leverage the ability to reconstruct the discharge from free surface measurements in order to infer the bottom topography.

In \cite{gessese2011reconstruction}, a direct algorithm for channel bathymetry was introduced, wherein the bed profile $b$ is reconstructed by a finite-difference discretization of \eqref{sh2} in which the roles of the free surface elevation $\zeta$ and the bed topography $b$ are interchanged within the momentum equation. Subsequently, this methodology was extended to two-dimensional shallow water in \cite{gessese2012direct} and using finite-element discretization in \cite{hajduk2020bathymetry}. All these algorithms presuppose steady-state flow and require knowledge of the free surface elevation, the inlet discharge, and the value of the bed topography at the upstream. To date, no rigorous proof of convergence has been provided. Furthermore, the dependency on the chosen initial conditions for these algorithms is not studied. As an alternative approach, the authors of \cite {gessese2011reconstruction} proposed a pseudo-analytic solution obtained by algebraic manipulation of \eqref{sh2} to yield an equation for $b$ \cite{gessese2011inferring}, which is then approximated using a first-order finite-difference scheme. Although this pseudo-analytic approach maintains the same prerequisite measurements as well as the steady-state assumption, it offers provable convergence, stability, and independence from the initial data.

One of the primary objectives of this paper is to develop a direct, fast, and stable analytical approach for identifying channel bathymetry in unsteady-state flow, utilizing only inlet discharge and free surface elevation measurements collected over a short time interval. 

\subsection{The inverse model}\label{inverseproblemintro}
In this part, we derive the inverse equation. The main idea relies on an algebraic reordering of system \eqref{sh2} and observing that the continuity equation of \eqref{sh2} allows for the reconstruction of the discharge at any instant $t$ over $(a_1,a_2)$ only by knowing $q(t,a_1)$ and $\partial_t \zeta(t,.)$, see Figure \ref{fig:inverse_domain}.

Consider an arbitrary open interval $I=(a_1,a_2)\subset \mathbb{R}$ within the open channel over which we aim to detect the bed topography. Assume the bottom profile $b$ is sufficiently smooth such that the shallow water system  \eqref{sh1} admits a unique solution $(\zeta,u)$ on $(0,t_f)\times I$, for some maximal time $t_f$ during which  $\zeta$ and $q$ remain sufficiently smooth. Local well-posedness of \eqref{sh1} on the unbounded domain $\mathbb{R}$ is proved in \cite{lannes2013water} and on the half-line  $\mathbb{R_+}$ in \cite{iguchi2021hyperbolic}. The local well-posedness on a bounded interval $I$  with non-vanishing velocity at $\partial I$ is established in  \cite{petcu2013one,huang2011one}. Throughout this paper, we assume that the system \eqref{sh1} admits a solution and that $(\zeta,q)$ is sufficiently smooth on $(0,t_f)\times I$.

Instead of directly dealing with system \eqref{sh2} as in \cite{gessese2011reconstruction}, we transform it into the following system:
\begin{equation}\label{sh3}
\begin{cases}
\displaystyle q(t,x) = q(t,a_1)-\displaystyle\int_{a_1}^x\partial_{t}\zeta(t,y)\,{\dd}y,\;\;\;\forall(t,x)\in(0,t_f)\times I,\\[6pt]
\partial_{x}\!\biggl(\frac{q^{2}}{h}\biggr) = -\,g\,h\,\partial_{x}\zeta \;-\;\partial_{t}q,\;\;\;\forall(t,x)\in(0,t_f)\times I.
\end{cases}
\end{equation}
As mentioned earlier, assuming  access to the inlet discharge $q(t,a_1)$ at an instant $t\in (0,t_f)$, we can observe that from the first equation of \eqref{sh3}, the discharge $q(t,x)$ at any point $x\in I$ can be directly  obtained by knowing $\partial_t\zeta(t,\cdot)$ over $(a_1,x)$. The hypothesis of a known inlet discharge is a standard assumption in the inverse problem of bottom topography detection \cite{gessese2011reconstruction,angel2024bathymetry,khan2021variational}, as well as in the control of the wave profile using inlet and outlet data \cite{sanders2000adjoint}.

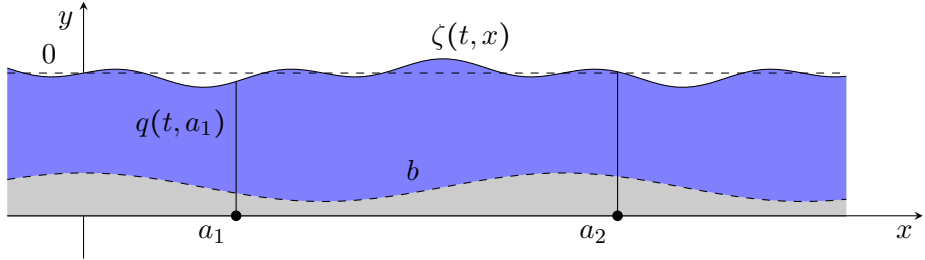
\begin{figure}[H]
\centering
\begin{tikzpicture}
\begin{axis}[
    width=0.8\textwidth,    
    height=5cm,
    axis lines = middle,
    xlabel = $x$,
    ylabel = $y$,
    xtick =  \empty,
    ytick = \empty,
    xmin = -1,
    xmax = 11,
    ymin = -3, 
    ymax = 15, 
    no markers,
    grid = none,
    every axis x label/.style={at={(current axis.right of origin)},anchor=north east},
    every axis y label/.style={at={(current axis.above origin)},anchor=north east},
    ]
    
    \addplot[black, domain=-1:10, samples=100,name path=Z] {sin(deg(x))*cos(deg(x*2))+10} node[pos=0.9, anchor=south west] {$ $} node[pos=0.5, anchor=south west] {$\zeta(t,x) $};
     \addplot[black, domain=-1:10, samples=100, name path=A] {0} ;
     
    \addplot[dashed, domain=-1:10, samples=100, name path=B] {cos(deg(x))+2} node[pos=0.9, anchor=south west] {$ $} node[pos=0.5, anchor=south east] {$ b $};
    
    \draw[black] (axis cs:2,0) -- (axis cs:2,9.45) node[pos=0.5, anchor=south east] {$q(t,a_1) $} ;
    \draw[black] (axis cs:7,0) -- (axis cs:7,10.1) node[pos=0.5, anchor=south west] {$\ $};
    
    \fill (axis cs:2,0) circle[radius=2pt] node[below left] {$a_1$};
    \fill (axis cs:7,0) circle[radius=2pt] node[below left] {$a_2$};
    
    \draw[dashed] (axis cs:-1,10) -- (axis cs:10,10)  node[pos=0.05, above] {0};
   
    
    
    \addplot[black!20] fill between[of=A and B];
    \addplot[blue!50] fill between[of=B and Z];

\end{axis}
\end{tikzpicture}
\caption{ Scheme for the inverse problem}
\label{fig:inverse_domain}
\end{figure}
To develop a one-shot approach for bathymetry detection in unsteady flows, we impose the following nondegeneracy assumption:
\begin{equation}\label{condition}
    \exists\; \beta>0,\;\;\exists\; t^{\star}\in(0,t_f),\;\; \forall\;x\in I,\quad q(t^{\star},x)\geq \beta.
\end{equation}
Then the inverse model can be written as follows:

\begin{equation}\label{sh4}
\begin{cases}
\displaystyle q(t^{\star},x) = q(t^{\star},a_1)-\displaystyle\int_{a_1}^x\partial_{t}\zeta(t^{\star},y)\,{\dd}y,\;\;\;\forall x\in I,\\[6pt]
\displaystyle \partial_tq(t^{\star},x) = \partial_t q(t^{\star},a_1)-\displaystyle\int_{a_1}^x\partial_{t}^2\zeta(t^{\star},y)\,{\dd}y,\;\;\;\forall x\in I,\\[6pt]
\displaystyle \partial_{x}\phi (t^{\star},x)
= -\,g\,\left(\partial_{x}\zeta \;q^2\right)(t^{\star},x)\;\frac{1}{\phi(t^{\star},x) }\;-\;\partial_{t}q(t^{\star},x), \;\;\;\forall x\in I, 
\end{cases}
\end{equation}
where $\phi =\frac{q^2}{h}$ and $t^{\star}$ is given by \eqref{condition}.
Under the assumption \eqref{condition}, the only free surface measurements required to solve the third equation in \eqref{sh4} are $\left(\zeta,\;\partial_t \zeta,\; \partial_{t}^2 \zeta\right)|_{\{t^{\star}\}\times I}$. It should be mentioned that these three measurements correspond precisely to those considered in \cite{vasan2013inverse} to detect the bottom profile $b$ using the Euler model when assuming a periodic velocity potential. After solving the equation for $\phi$ in \eqref{sh4}, the bottom topography can be computed as follows:
\begin{equation}\label{bottomsolution}
b(x)=\zeta(t^{\star},x)-\frac{q^2(t^{\star},x)}{\phi(t^{\star},x)},\quad \forall x\in I.\\[6pt]
\end{equation}
\begin{remark}\label{abouttheinversemodel}
    Although the work adopts the context of measurements in a single instant, it can be applied to any time interval $(0,t_f)$. If we have $\zeta$ and $q(\cdot,a_1)$ over $(0,t_f)$, then we can reconstruct $q(t,x)$ over $(0,t_f)\times I$ via the continuity equation (first equation of \eqref{sh3}) and numerically computing its time derivative, or directly evaluating $\partial_t q(t,x)$ from the second equation of \eqref{sh4}. Finally, by choosing any $t^{\star}$ such that the discharge is strictly positive  \eqref{condition}, we approximate the inverse equation on $\phi$ in the above system \eqref{sh4} to infer $b$ \eqref{bottomsolution}.
    \end{remark}
Compared to optimization-based approaches and direct steady-state schemes, the major difference is the nondegeneracy assumption \eqref{condition}. However, this requirement is naturally met in open channels, where the fluid flows from upstream ($a_1$) to downstream ($a_2$). For the moment, assuming a globally supercritical flow: for any $(t,x)\in(0,t_f)\times I$, we have
\begin{equation}\label{supercritical}
u(t,x)>\sqrt{g\,h(t,x)}.
\end{equation}
Then, the inequality  \eqref{condition} holds globally, since $q=h\,u$ and $h>\varepsilon$ for some positive constant $\varepsilon$. The same reasoning applies for flow with strong supercritical-type properties \cite{huang2011one} and for the 
strengthened supercritical-type conditions considered in \cite{kounadis2020galerkin}.

The system \eqref{sh4}, along with knowledge of $b(a_1)$, can be seen as the natural extension of the steady-state case considered in \cite{gessese2011reconstruction,gessese2011inferring}. Moreover, based on \eqref{sh4}, where the friction term is neglected, we can derive an exact analytic solution for $b$ as follows:\\
Assume the steady state ($\partial_t =0$). Then \eqref{sh4} reduces to
\begin{equation*}
\begin{cases}
\displaystyle q(t,x) = q(0,a_1)=q,\;\; \forall x\in I,\\[6pt]
\displaystyle \phi(x)\,\partial_{x}\phi (x)
= -\,g\,\left(\partial_{x}\zeta \;q^2\right)(x), \;\; \forall x\in I, 
\end{cases}
\end{equation*}
straightforward integration over $(a_1,x)$ for any $x\in I$, yields

\begin{equation}\label{sh5}
\begin{cases}
\displaystyle \frac{1}{h^2(x)} = \frac{1}{h^2(a_1)}-\frac{2 g}{q^2}(\zeta(x)-\zeta(a_1)),\\[6pt]
\displaystyle b=\zeta-h.
\end{cases}
\end{equation}
In addition, if we assume that the bottom topography is known at a second point $x_{\ast}$ in $\{x\in I,\;\; \zeta(x)\neq \zeta(a_1)\}$ (for example, if possible, $a_2$). Then the inlet discharge $q$ in \eqref{sh5} can be replaced by:

\begin{equation}\label{simpliciation for discharge}
q^2=-2g\,\frac{\left(\zeta(x_{\ast})-\zeta(a_1)\right)\,h^2(a_1) h^2(x_{\ast})}{h^2(a_1)-h^2(x_{\ast})},
\end{equation}
which consequently yields an analytic bathymetry reconstruction independent of the inlet discharge.

We have introduced our analytic approach for the inverse problem of bathymetry detection in shallow water for both steady and unsteady cases.  Although strictly positive discharge over the domain holds naturally in the steady and supercritical cases, in the next section, we demonstrate that this condition remains valid for subcritical and transcritical flows.

\section{ Sufficient Conditions and Stability}\label{sec:guaranteeandstability}
In this section, we begin by presenting the main Theorem \ref{subcriticaltheorem},  which delivers a sufficient condition for the nondegeneracy requirement \eqref{condition} for subcritical flow cases. This criterion is based on an inequality linking the inlet data, outlet data, and bed profile variation with the characteristic propagation speeds. As a direct consequence of this theorem, we derive Corollary \ref{transflow}, which introduces a second sufficient condition for \eqref{condition} that allows transcritical flows. Importantly, this alternative condition \eqref{condition on bottom modified} depends only on inlet data, which matches the focus of this paper. Crucially, it reveals that bed profile variation is controlled by the right‑going characteristic speed. Finally, in Subsection \ref{stabilityanalysis}, we establish a Lipschitz stability estimate for our reconstruction method and illustrate the role of free surface smoothness in error amplification.
\subsection{Nondegeneracy Sufficient Conditions}
In the previous Section \ref{inverseproblemintro}, the inverse approach was introduced under the nondegeneracy of discharge  \eqref{condition}. Although this assumption is automatically fulfilled for supercritical flows, Theorem \ref{subcriticaltheorem} below establishes that it continues to hold in the strong subcritical cases considered in \cite{iguchi2021hyperbolic,kounadis2020galerkin}. Specifically, we assume that there exist two positive constants $c_1,\,c_2>0$ such that for all $(t,x)\in (0,t_f)\times I $ , we have 
\begin{equation}\label{subcriticalflow}
\displaystyle \lambda_1=u(t,x)+\sqrt{g\,h(t,x)}\geq c_1\quad\text{and}\quad \lambda_2=u(t,x)-\sqrt{g\,h(t,x)}\leq - c_2.\\[6pt]
\end{equation}

\begin{theorem}\label{subcriticaltheorem}
Let $\varepsilon>0$ such that $h(t,x)\geq \varepsilon,\; \forall (t,x)\in(0,t_f)\times I$. Assume that the bed profile $b\in C^1(I)$ and that the fluid flow satisfies the hypothesis \eqref{subcriticalflow}. If there exists $\rho >0$, such that for any $t\in \left(\frac{a_2-a_1}{c_0},t_f\right)$, where $c_0=\min(c_1,c_2)$, the following condition holds:
\begin{equation}\label{condition on bottom}
\displaystyle u(s,a_1)+u(\tau,a_2)\geq \rho -2\sqrt{g}\left(\sqrt{h(s,a_1)}-\sqrt{h(\tau,a_2)}\right)+g\left(\frac{1}{c_1}+\frac{1}{c_2}\right)\displaystyle\sum_{i\geq 1} b(\gamma_{i+1})-b(\gamma_{i}),\quad\forall (s,\tau)\in (0,t), \\[6pt]
\end{equation}
where $(\gamma_i)_i $ are such that $b'(\gamma_i)=0$ and  $\bigcup_{i\geq 1} (\gamma_i,\gamma_{i+1})=\left\{x\in I,\;\; b'(x)>0\right\}$. Then,  for any $t\in \left(\frac{a_2-a_1}{c_0},t_f\right)$, we have 
\begin{equation}\label{donesubcretical}
\displaystyle q(t,x)\geq \varepsilon\,\frac{\rho }{2},\quad \forall x \in I.
\end{equation}
\end{theorem}
\begin{proof}
The proof exploits the fact that information within the domain can be obtained from the upstream and downstream boundaries (see Figure \ref{fig:inverse_domain}) via the characteristic speeds  $\lambda_1 $  and $\lambda_2$ defined in \eqref{subcriticalflow}.

We first recast system \eqref{sh2} in its equivalent balance-law formulation:
\begin{equation}\label{thp1}
\begin{cases}
\displaystyle \partial_t R_1+\lambda_1\,\partial_x R_1=-g\, b',\;\; \text{in}\;(0,t_f)\times I,\\[6pt]
\displaystyle \partial_t R_2+\lambda_2\,\partial_x R_2=-g\,b',\;\; \text{in}\;(0,t_f)\times I,
\end{cases}
\end{equation}
where $R_1$ and $R_2$ are the Riemann invariants, defined as follows:
$$R_1=u+2\sqrt{g\,h}\quad\text{and}\quad R_2=u-2\sqrt{g\,h}.$$
Assuming that the flow hypothesis \eqref{subcriticalflow} and the condition \eqref{condition on bottom} are satisfied, fix an arbitrary point $(t,x)\in\left(\frac{a_2-a_1}{c_0},t_f\right)\times I$. We now proceed to establish \eqref{donesubcretical}. By applying the method of characteristics, the system \eqref{thp1} can be expressed in its characteristic form as follows

\begin{equation}\label{thp2}
\begin{cases}
\displaystyle \partial_s R_1 \left(s,X_1(s)\right)=-g\,b'\left(X_1(s)\right), \;\; \forall s\in(0,t),\\[6pt]
\displaystyle X_1'(s)=\lambda_1\left(s,X_1(s)\right),\;\;\forall s\in(0,t);\;\; X_1(t)=x,
\end{cases}
\end{equation}
\begin{equation}\label{thp3}
\begin{cases}
\displaystyle \partial_s R_2 \left(s,X_2(s)\right)=-g\,b'\left(X_2(s)\right), \;\; \forall s\in(0,t),\\[6pt]
\displaystyle X_2'(s)=\lambda_2\left(s,X_2(s)\right),\;\;\forall s\in(0,t);\;\; X_2(t)=x.
\end{cases}
\end{equation}
Integrating the second equations of systems \eqref{thp2} and \eqref{thp3}, we obtain
\begin{equation*}
\begin{cases}
\displaystyle x-X_1(\tau)=\displaystyle\int_{\tau}^t\,\lambda_1\left(s,X_1(s)\right) \,{\dd}s,\;\; \forall\tau\in(0,t),\\[6pt]
\displaystyle x-X_2(\tau)=\displaystyle\int_{\tau}^t\,\lambda_2\left(s,X_1(s)\right)\,{\dd}s,\;\; \forall\tau\in(0,t).
\end{cases}
\end{equation*}
Applying hypothesis \eqref{subcriticalflow} yields the following inequalities

\begin{equation*}
\begin{cases}
\displaystyle X_1(\tau)\leq x-c_0(t-\tau),\;\; \forall\tau\in(0,t),\\[6pt]
\displaystyle X_2(\tau)\geq x+c_0(t-\tau),\;\; \forall\tau\in(0,t).
\end{cases}
\end{equation*}
Using $t>\frac{a_2-a_1}{c_0}$, we get
\begin{equation}\label{thp4}
\begin{cases}
\displaystyle X_1(0)\leq x-c_0t<a_1,\\[6pt]
\displaystyle X_2(0)\geq x+c_0t>a_2.
\end{cases}
\end{equation}
By the mean-value theorem, since $X_1$ is strictly increasing on $(0,t)$ and $X_2$ is strictly decreasing on $(0,t)$, along with \eqref{thp4}, there exist unique $\tau_1$ and $\tau_2$ in $(0,t)$  that depend on $x$ and $t$, satisfying
$$X_1(\tau_1)=a_1\quad\text{and}\quad X_2(\tau_2)=a_2.$$
Integrating the first equations of systems \eqref{thp2} and \eqref{thp3} over $(\tau_1,t)$ and $(\tau_2,t)$, respectively, we arrive at
\begin{equation}\label{indepencyofintialdata}
\begin{cases}
\displaystyle R_1(t,x)-R_1(\tau_1,a_1)=-g\displaystyle\int_{\tau_1}^t b'\left(X_1(s)\right)\,{\dd}s,\\[6pt]
\displaystyle R_2(t,x)-R_2(\tau_2,a_2)=-g\displaystyle\int_{\tau_2}^t b'\left(X_2(s)\right)\,{\dd}s.
\end{cases}
\end{equation}
In the following steps of the proof, we focus only on the Riemann invariant $R_1$; analogous reasoning applies to $R_2$. We have 
\begin{equation}\label{R1expr}
R_1(t,x)=u(\tau_1,a_1)+2\sqrt{g\,h(\tau_1,a_1)}\;\,-g\displaystyle\int_{\tau_1}^t b'\left(X_1(s)\right)\,{\dd}s.
\end{equation}
On the other hand, using the change of variable $r=X_1(s)$, the right-hand integral that involves $b$ in the above equation \eqref{R1expr} can be rewritten as follows
\begin{equation}\label{bottomintegral}
\begin{aligned}
\displaystyle\int_{\tau_1}^t b'\left(X_1(s)\right)\,{\dd}s&=\displaystyle\int_{a_1}^x b'(r)\;\frac{1}{\lambda_1\left(X_1^{-1}(r),r\right)}\,{\dd}r\\
&=\displaystyle\int_{\left\{(a_1,x),\;\;b'>0\right\}} \frac{b'}{\lambda_1}\,+ \displaystyle\int_{\left\{(a_1,x),\;\;b'<0\right\}} \frac{ b'}{\lambda_1}\\
&\leq \displaystyle\int_{\left\{(a_1,x),\;\;b'>0\right\}} \frac{ b'}{\lambda_1}\\
&\leq \frac{1}{c_1}\displaystyle\int_{\left\{I,\;\;b'>0\right\}}  b'.
\end{aligned}
\end{equation}
Substituting the above inequality \eqref{bottomintegral} into \eqref{R1expr}, we obtain the following inequality
\begin{equation}\label{R1exprineq}
R_1(t,x)\geq u(\tau_1,a_1)+2\,\sqrt{g\,h(\tau_1,a_1)}\;\,-\frac{g}{c_1}\displaystyle\int_{\left\{I,\;\;b'>0\right\}} b'.
\end{equation}
Following similar arguments, with attention to the fact that $X_2$ is a strictly decreasing function over $(0,t)$, we derive the following lower bound for $R_2$
\begin{equation}\label{R2exprineq}
R_2(t,x)\geq u(\tau_2,a_2)-2\,\sqrt{g\,h(\tau_2,a_2)}\;\,-\frac{g}{c_2}\displaystyle\int_{\left\{I,\;\;b'>0\right\}}  b'.
\end{equation}
Summing the above inequalities \eqref{R1exprineq} and \eqref{R2exprineq} side by side, we obtain
\begin{equation*}
2\,u(t,x)\geq u(\tau_1,a_1)+u(\tau_2,a_2)+2\sqrt{g}\,\left(\sqrt{h(\tau_1,a_1)}-\sqrt{h(\tau_2,a_2)}\right)\;\,-g\left(\frac{1}{c_1}+\frac{1}{c_2}\right)\displaystyle\int_{\left\{I,\;\;b'>0\right\}}  b'.
\end{equation*}
Since $ b'$ is continuous on $I$, then $\left\{I,\;\;b'>0\right\}$ is an open set of $\mathbb{R}$. Therefore,  the above integral of $ b'$ can be written as the Riemann-sum in \eqref{condition on bottom}. Hence, the condition \eqref{condition on bottom} implies
$$2\,u(t,x)\geq \rho .$$
Finally,
$$q(t,x)=h(t,x)u(t,x)\geq \varepsilon\,\frac{\rho }{2},$$
which completes the proof.
\end{proof}
Before presenting a sufficient condition that utilizes only the inlet data to guarantee the nondegeneracy assumption \eqref{condition}, the following remarks are crucial.
\begin{remark}\label{r2}
Under the hypothesis \eqref{subcriticalflow} and regularity conditions on the bottom profile, initial data, and boundary data of the system \eqref{sh2}, a result of local existence in $\mathbb{W}^m(t_f) = \cap_{j=0}^m C^j([0,t_f],H^{m-j}(\mathbb{R}_+)^2)$ for $m \ge 2$ is established in \cite[Theorem~2.25]{iguchi2021hyperbolic}. Moreover,  a direct continuity argument implies that, if there exists $ \bar t \in (0,t_f)$ and $ \beta > 0$ such that $ \,  q(\bar t, x) \ge \beta, \, \forall x  \in I$, then there exists a neighborhood $V_{\bar t}$ of  $\bar t$ in $(0,t_f)$ such that  $ \,  q(t,x) \ge \beta/2, \, \forall (t,x)  \in V_{\bar t}  \times I$.
\end{remark}

\begin{remark}\label{r1}
    Theorem \ref{subcriticaltheorem} remains valid if we only require the flow hypothesis \eqref{subcriticalflow} and the sufficient condition \eqref{condition on bottom} on a shifted intervals   $(t_s,t_f)$ and $(t_s,t)$, respectively, where $t_f>t_s+\frac{a_2-a_1}{c_0}$.  Following the same steps in the proof, we then obtain \eqref{donesubcretical} for any $t\in\left(t_s+\frac{a_2-a_1}{c_0},\,t_f\right)$. This extension is particularly important when hypothesis \eqref{subcriticalflow} does not hold in the neighborhood of the initial data $t=0.$ See {\bf \em Test 2}.
\end{remark}
 
\begin{remark}\label{r3}
The sufficient condition \eqref{condition on bottom} that guarantees a positive discharge \eqref{condition} requires considering bed profiles with a few variations (see Figure 1.10 in \cite{lannes2013water}) and a maximum amplitude $\max b$ remains small relative to the wave speeds $c_1,\;c_2$. In particular, this scenario meets the cases of detecting low-amplitude bed obstacles such as those considered in \cite{tam2015predicting,angel2024bathymetry,gessese2011reconstruction,khan2021variational}.
\end{remark}

\begin{remark}\label{problem in condition}
Theorem \ref{subcriticaltheorem} provides a general condition to guarantee \eqref{donesubcretical}. However, in open channels where the fluid flows from upstream $a_1$ to downstream $a_2$, the left-going wave speed $c_2$ can be small. Hence, the term involving the bottom variation in \eqref{condition on bottom} may result in a large value if $\max b$ is not sufficiently small. Furthermore, the well-posedness and observation time  $t_f$ should be large enough. 
\end{remark}

The following result illustrates a sufficient condition requiring only inlet data. This hypothesis is motivated by Remark \ref{problem in condition} and by the aim of considering transcritical flow as well.  We assume that there exist positive constant $c_1>0$ such that for all $(t,x)\in (0,t_f)\times I $ , we have 
\begin{equation}\label{transcriticallflow}
\displaystyle \lambda_1(t,x)\geq c_1,\\[6pt]
\end{equation}
where $\lambda_1$ is given by \eqref{subcriticalflow}.
\begin{corollary}\label{transflow}
Let $\varepsilon>0$ such that $h(t,x)\geq \varepsilon,\; \forall (t,x)\in(0,t_f)\times I$. Assume that the bed profile $b\in C^1(I)$ and that the fluid flow satisfies the hypothesis \eqref{transcriticallflow}. If there exists $\rho >0$, such that for any $t\in \left(\frac{a_2-a_1}{c_1},t_f\right)$, the following condition holds:
\begin{equation}\label{condition on bottom modified}
\displaystyle u(s,a_1)\geq \rho  -2\sqrt{g}\left(\sqrt{h(s,a_1)}-\sqrt{\max_{(0,t_f)\times I}h}\right)+\frac{g}{c_1}\displaystyle\sum_{i\geq 1} b(\gamma_{i+1})-b(\gamma_{i}),\quad\forall s\in (0,t), \\[6pt]
\end{equation}
where  $(\gamma_i)_i $ are such that $b'(\gamma_i)=0$ and $\bigcup_{i\geq 1} (\gamma_i,\gamma_{i+1})=\left\{x\in I,\;\;  b'(x)>0\right\}$. Then,  for any $t\in \left(\frac{a_2-a_1}{c_1},t_f\right)$, we have 
\begin{equation}\label{donesubcreticalmod}
\displaystyle q(t,x)\geq \varepsilon\,\rho ,\quad \forall x \in I.
\end{equation}
\end{corollary}

\begin{proof}
    The result follows directly from the proof of Theorem \ref{subcriticaltheorem}, in particular from inequality \eqref{R1exprineq}.
\end{proof}

\begin{remark}
    Compared to condition \eqref{condition on bottom}, condition \eqref{condition on bottom modified} is stated in terms of the difference between the inlet depth and the maximum depth, rather than the change of flow depth between the inlet and the outlet \eqref{condition on bottom}. However, both criteria quantify the relative change of the flow depth from the inlet depth. Furthermore, numerical experiments in Section \ref{sec:numericalexpr} confirm that the sufficient condition \eqref{condition on bottom modified} holds, including the hydraulic falls scenarios \cite{tam2015predicting}.
\end{remark}

\begin{remark}
   Although the sufficient condition  \eqref{condition on bottom modified} imposes constraints on the smoothness and amplitude of the bathymetry, it should be mentioned that this term is controlled by the right-going wave speed $c_1$.
\end{remark}

We note that the conclusions of Remarks \ref{r1}-\ref{problem in condition} extend directly to the above Corollary \ref{transcriticallflow}.

Two sufficient conditions \eqref{condition on bottom}, \eqref{condition on bottom modified} that ensure the validity of the direct-approach assumption \eqref{condition} are introduced. These conditions rely on an inequality that balances the inlet or outlet data with the smoothness of the bottom and the speed of wave propagation. Alternatively, assumption \eqref{condition} can be verified directly by discretizing the first equation in  \eqref{sh4} and employing time-series data of inlet discharge and free surface elevation over the interval $(0,t_f)$, see Remark \ref{abouttheinversemodel}. We now proceed to analyze the sensitivity of our analytic reconstruction method.

\subsection{Sensitivity analysis}\label{stabilityanalysis}
We demonstrate that the mapping from the free surface measurements \eqref{surafcemeasurement} to the reconstructed bed profile $b$ is Lipschitz continuous in the $L^1-$ norm.

Assuming that one of the hypotheses \eqref{subcriticalflow}–\eqref{condition on bottom} or \eqref{transcriticallflow}–\eqref{condition on bottom modified} is satisfied. Then, we have \eqref{condition}: there exists $\beta>0$ and $t^{\star}\in(0,t_f)$ such that
$$q(t^{\star},x)\geq \beta,\quad \forall x\in I.$$ 
Since $\left(q, \partial_t q\right)|_{\{t^{\star}\}\times I}$ are linked to $\left(\partial_t\zeta,\partial_t^2\zeta\right)|_{\{t^{\star}\}\times I}$ by the inlet discharge through the first and second equations in \eqref{sh4}, respectively, without loss of generality, we investigate the sensitivity using the measurements
\begin{equation}\label{Ms}
\displaystyle M_s=\left(b(a_1),\left(\zeta,q,\partial_t q\right)|_{\{t^{\star}\}\times I}\right).\\[6pt]
\end{equation}
Let $\widetilde{M_s}$ denotes a perturbed measurement of $M_s$, and let $\widetilde{\phi}$ be the corresponding error solution of \eqref{sh4}. Then, we have the following two systems associated with $M_s$ and $\widetilde{M_s}$, respectively:
\begin{equation}\label{exactsystem}
\begin{cases}
\displaystyle \partial_x \phi(t^{\star},.)=-g \frac{\partial_x \zeta(t^{\star},.)\, q^2(t^{\star},.)}{\phi(t^{\star},.)}\,-\,\partial_t q(t^{\star},.),\quad \text{in}\; I,\\[6pt]
\displaystyle \phi(t^{\star},a_1)=\frac{q^2(t^{\star},a_1)}{h(t^{\star},a_1)},
\end{cases}
\end{equation}
\begin{equation}\label{errorsystem}
\begin{cases}
\displaystyle \partial_x \widetilde{\phi}(t^{\star},.)=-g \frac{\partial_x \widetilde{\zeta}(t^{\star},.)\, (\widetilde{q}(t^{\star},.))^{2}}{\widetilde{\phi}(t^{\star},.)}\,-\,\partial_t\widetilde{ q}(t^{\star},.),\quad \text{in}\;I,\\[6pt]
\displaystyle \widetilde{\phi}(t^{\star},a_1)=\frac{(\widetilde{q})^{2}(t^{\star},a_1)}{\widetilde{h}(t^{\star},a_1)}.
\end{cases}
\end{equation}

Needless to say that throughout this section, the solution is assumed to be sufficiently regular to make sense of the terms of the above systems \eqref{exactsystem} and \eqref{errorsystem}. Such regularity is assured as soon as we place ourselves in the conditions of Theorem 2.25 in \cite{iguchi2021hyperbolic} since then the solution belongs to $\mathbb{W}^m(t_f), \, m \ge 2$.

For simplicity,  we omit in the following the explicit dependence on $t^{\star}$. However, it should be mentioned that the analysis refers to the fixed instant $t^{\star}$. In addition, the perturbations $\widetilde{q}$ and $\partial_t\widetilde{q}$ arise directly from the  errors $\partial_t \widetilde{\zeta}$ and $\partial_t^2 \widetilde{\zeta}$ via the first and the second equations in \eqref{sh4}, respectively. In particular, $\partial_t\widetilde{q}$ is defined by 
$$\partial_t\widetilde{q}(x):=\partial_t\widetilde{q}(t^{\star},x) = \partial_t \widetilde{q}(t^{\star},a_1)-\displaystyle\int_{a_1}^x\partial_{t}^2\widetilde{\zeta}(t^{\star},y)\,{\dd}y,\;\;\;\forall x\in I.$$
Using this simplified notation, we have the following stability estimate.
\begin{proposition}\label{propositionstability}
Let $b,\;\widetilde{b}\in W^{1,1}(I)$, $\zeta,\,\widetilde{\zeta}\in W^{1,\infty}(I)$;  $q,\,\widetilde{q},\,\partial_t q,\,\partial_t\widetilde{ q} \in W^{1,1}(I)$, and  $\varepsilon>0$ such that $h,\,\widetilde{h}\geq \varepsilon$ . Under the assumption \eqref{condition}, we have Lipschitz stability for the detection of the bed profile using $M_s$ \eqref{Ms}. Precisely, there exist positive constants $C_1,\,C_2,\,C_3,\,C_4>0$, such that
\begin{equation}\label{stability}
\displaystyle \left\Vert b-\widetilde{b}\right\Vert_{L^1(I)}\leq C_{1,2} \,E+\frac{1}{\beta^2}\left\Vert \widetilde{h}(q+\widetilde{q})\right\Vert_{L^{\infty}(I)}\,\left\Vert q-\widetilde{q}\right\Vert_{L^{1}(I)}+\left\Vert \zeta-\widetilde{\zeta}\right\Vert_{L^1(I)},\\[6pt]
\end{equation}
where
\begin{equation}\label{C11}
    C_{1,2}=\frac{1}{C_1}\biggl(e^{C_2(a_2-a_1)}-1\biggr),
\end{equation}

\begin{equation}\label{E}
\begin{aligned}
E&= \frac{q^2(a_1)}{h(a_1)\widetilde{h}(a_1)}\Bigl(\left|b(a_1)-\widetilde{b}(a_1)\right|+\left|\zeta(a_1)-\widetilde{\zeta}(a_1)\right|\Bigr)+\frac{q(a_1)+\widetilde{q}(a_1)}{\widetilde{h}(a_1)}\left|q(a_1)-\widetilde{q}(a_1)\right|\\
&\quad\quad+C_3 \left\Vert q-\widetilde{q}\right\Vert_{L^{1}(I)}+C_4 \left\Vert \partial_x(\zeta-\widetilde{\zeta})\right\Vert_{L^{1}(I)}+ \left\Vert \partial_t q-\partial_t\widetilde{ q}\right\Vert_{L^{1}(I)},\\
\end{aligned}
\end{equation}
and 
\begin{equation}\label{stabilitycanst}
\begin{cases}
\displaystyle C_1=g\left\Vert \partial_x\zeta\right\Vert_{L^{\infty}(I)},\\[6pt]
\displaystyle C_2=\frac{g}{\beta^2}\left\Vert \partial_x\zeta\right\Vert_{L^{\infty}(I)}\left\Vert h\,\widetilde{h}\right\Vert_{L^{\infty}(I)},\\[6pt]
\displaystyle C_3=\frac{g}{\beta^2}\left\Vert \widetilde{h}\right\Vert_{L^{\infty}(I)}\,\left\Vert \left(\partial_x\widetilde{\zeta}\right)\left(q+\widetilde{q}\right)\right\Vert_{L^{\infty}(I)},\\[6pt]
\displaystyle C_4=\frac{g}{\beta^2}\left\Vert \widetilde{h}\right\Vert_{L^{\infty}(I)}\,\left\Vert q^2\right\Vert_{L^{\infty}(I)}.\\[6pt]
\end{cases}
\end{equation}
\end{proposition}
\begin{proof}
    For ease of reading, the detailed proof is moved to Appendix ~\ref {sec:A}.  Its core argument applies Grönwall’s inequality \cite{evans2022partial} to establish an $L^1$ stability bound under perturbations of the variable coefficients and initial Cauchy data of the ODE \eqref{exactsystem}. 
\end{proof}
We end this section with the following remark.
\begin{remark}\label{stability problm}
    In the above stability result \eqref{stability}, the constant $C_{1,2}$ \eqref{C11} is an exponential constant that depends on the Lipschitz constant $\bigl\|\partial_x\zeta\bigr\|_{L^{\infty}(I)}$.  Specifically, as $\bigl\|\partial_x\zeta\bigr\|_{L^{\infty}(I)}$ grows due to roughness or noise in the free surface data $\zeta$, this stability constant increases sharply. This reflects the impact of the smoothness of the free surface on the detection of bed topography. Nevertheless, in our numerical tests, this issue is resolved by applying a smoother to prevent the introduction of large gradient values of the free surface in the inverse model \eqref{sh4}; see the details in  {\bf \em Test 5}.
\end{remark}

\section{Numerical results}\label{sec:numericalexpr}
We divide this section into three subsections. In the first Subsection \ref{discritizationofforwared}, we briefly recall the finite‑volume discretization \cite{kurganov2007second} that we employ to approximate the forward model \eqref{sh2} and thus produce the required free surface measurements \eqref{surafcemeasurement}. In Section \ref{inversedisc}, we provide the detailed algorithm for approximating the inverse model \eqref{sh4} using a second-order finite-difference scheme. Finally, Section \ref{testn} presents five numerical tests that confirm the validity of the assumption of nondegeneracy \eqref{condition} and its sufficient conditions \eqref{condition on bottom} -\eqref{condition on bottom modified}. Furthermore, these tests demonstrate the high accuracy of the method and stability under noisy measurements. The Python script accompanying this section is given in \cite{codepyhon}.

\subsection{Discretization of the forward system}\label{discritizationofforwared}
We adopt the finite-volume scheme, detailed in \cite{kurganov2007second,kurganov2002central}, which simultaneously preserves steady states and non-negativity of fluid depth $h$ by incorporating a correction for dry regions. For further details, see the primary reference for this subsection \cite{kurganov2007second}. A corresponding Python script is provided in \cite{codepyhon} (file name: needed functions for code).

Based on \cite{kurganov2007second}, we first express the system \eqref{sh2} in the following equivalent form:
\begin{equation}\label{sh6}
\begin{cases}
\displaystyle \partial_{t}\zeta + \partial_{x}q = 0,\\[6pt]
\displaystyle \partial_{t}q + \partial_x\left(\frac{q^2}{\zeta-b}+\frac{g}{2}\left(\zeta-b\right)^2\right)= -g\left(\zeta-b\right)\, b'.
\end{cases}
\end{equation}
Let $U=(\zeta,q)^T$ and define 
$$F(U,b)=\left(q,\frac{q^2}{\zeta-b}+\frac{g}{2}\left(\zeta-b\right)^2\right)^T\quad\text{and}\quad S=\left(0,-g\left(\zeta-b\right)\, b'\right)^T.$$
Then, the above system \eqref{sh6} can be written in the following vector form
\begin{equation*}
\displaystyle \partial_{t}U + \partial_{x}F = S.\\[6pt]
\end{equation*}
For the sake of simplicity, we denote the grid points by $x_{\alpha}=\alpha\Delta x$, where $\Delta x$ is a chosen spatial step, and we denote finite-volume cells by $I_j=[x_{j-\frac{1}{2}},x_{j+\frac{1}{2}}]$. Then, the main steps of the second-order central-upwind scheme \cite{kurganov2007second} can be summarized as follows.\\
The semi-discrete system of ODEs is given by
\begin{equation}\label{algo1}
\displaystyle \frac{{\dd}}{{\dd} t}\overline{U}_j(t)=-\frac{H_{j+\frac{1}{2}}(t)-H_{j-\frac{1}{2}}(t)}{\Delta x}+\overline{S}_j(t),\quad \forall t\in (0,t_f),
\end{equation}
where $$\overline{U}_j(t)\approx \frac{1}{\Delta x}\displaystyle\int_{I_j}U(t,x)\,{\dd}x, \quad \forall t\in (0,t_f).$$
To allow for discontinuous bathymetry, we approximate the bottom profile using a locally piecewise-linear approximation over each cell:
\begin{equation}\label{algo2}
\displaystyle \breve{b}(x)=b_{j-\frac{1}{2}}+\left(b_{j+\frac{1}{2}}-b_{j-\frac{1}{2}}\right)\frac{x-x_{j-\frac{1}{2}}}{\Delta x},\quad \quad  x\in I_j.
\end{equation}
where
\begin{equation}\label{algo3}
\displaystyle b_{j+\frac{1}{2}}:=\frac{b(x_{j+\frac{1}{2}}+0)+b(x_{j+\frac{1}{2}}-0)}{2}.
\end{equation}
Here,  $x_{j+\frac{1}{2}}+0$ and $x_{j+\frac{1}{2}}-0$ are the right and left limits at $x_{j+\frac{1}{2}}$, respectively \cite{kurganov2007second}. Then
\begin{equation}\label{algo4}
\displaystyle b_j:=\breve{b}(x_j)= \frac{1}{\Delta x}\displaystyle\int_{I_j} \breve{b}(x)\,{\dd}x=\frac{b_{j+\frac{1}{2}}+b_{j-\frac{1}{2}}}{2}.
\end{equation}
For ease of reading, we omit the explicit time dependence ($t$) in the following steps. Using \eqref{algo3} and \eqref{algo4}, the terms in \eqref{algo1} are given by
\begin{equation}\label{algo5}
\displaystyle \overline{S}_j^{(2)}\approx -g\left(\overline{\zeta}_j-b_j\right)\frac{b_{j+\frac{1}{2}}-b_{j-\frac{1}{2}}}{\Delta x}
\end{equation}
and 
\begin{equation}\label{algo6}
\begin{aligned}
H_{j+\frac{1}{2}}&=\frac{A^{+}_{j+\frac{1}{2}}F\left(U^{-}_{j+\frac{1}{2}},b_{j+\frac{1}{2}}\right)-A^{-}_{j+\frac{1}{2}}F\left(U^{+}_{j+\frac{1}{2}},b_{j+\frac{1}{2}}\right)}{A^{+}_{j+\frac{1}{2}}-A^{-}_{j+\frac{1}{2}}} \\
&\quad\quad+\frac{A^{+}_{j+\frac{1}{2}}\,A^{-}_{j+\frac{1}{2}}}{A^{+}_{j+\frac{1}{2}}-A^{-}_{j+\frac{1}{2}}}\left(U^{+}_{j+\frac{1}{2}}-U^{-}_{j+\frac{1}{2}}\right).
\end{aligned}
\end{equation}
Here, 
\begin{equation}\label{algo7}
U^{\pm}_{j+\frac{1}{2}}=\overline{U}_{j+\frac{1}{2}\pm\frac{1}{2}}\mp\frac{\Delta x}{2}\left(\partial_x U\right)_{j+\frac{1}{2}\pm\frac{1}{2}},
\end{equation}
where $(\partial_x U)_j$ is computed using the generalized minmod limiter 
\begin{equation}\label{algo8}
(\partial_x U)_j=\minmod\left(\theta\frac{\overline{U}_j-\overline{U}_{j-1}}{\Delta x},\frac{\overline{U}_{j+1}-\overline{U}_{j-1}}{2\Delta x},\theta\frac{\overline{U}_{j+1}-\overline{U}_{j}}{\Delta x}\right), \quad \theta\in[1,2],
\end{equation}
defined by
\[
\minmod(z_1,z_2,z_3) =
\begin{cases}
\min\bigl(z_1,z_2,z_3\bigr), & \text{if } z_i>0,\, \forall i\in\{1,2,3\},\\
\max\bigl(z_1,z_2,z_3\bigr), & \text{if } z_i<0,\, \forall i\in\{1,2,3\},\\
0, & \text{otherwise}.
\end{cases}
\]
In \eqref{algo6}, the local speeds of  propagation at each interface $x_{j+\frac{1}{2}}$ are defined as
\begin{equation}\label{algo9}
\begin{cases} 
    \displaystyle  A^{+}_{j+\frac{1}{2}}=\max\left\{u^{+}_{j+\frac{1}{2}}+\sqrt{g\,h^{+}_{j+\frac{1}{2}}},\;u^{-}_{j+\frac{1}{2}}+\sqrt{g\,h^{-}_{j+\frac{1}{2}}},\;0\right\}   ,\\[6pt]
    \displaystyle  A^{-}_{j+\frac{1}{2}}=\min\left\{u^{+}_{j+\frac{1}{2}}-\sqrt{g\,h^{+}_{j+\frac{1}{2}}},\;u^{-}_{j+\frac{1}{2}}-\sqrt{g\,h^{-}_{j+\frac{1}{2}}},\;0\right\}.
\end{cases}
\end{equation}
It should be mentioned that the central‑upwind algorithm outlined in steps \eqref{algo1}-\eqref{algo9} does not inherently preserve positivity of the local depth values $h_{j+\frac{1}{2}}$. In \cite{kurganov2007second}, the authors introduced corrections to steps \eqref{algo7}-\eqref{algo8} to ensure $h_{j+\frac{1}{2}}>0$. For brevity, the positivity‐preserving details are omitted; we refer the reader to the process (2.15)–(2.21) in \cite{kurganov2007second} or the Python script of this work  \cite{codepyhon}.

For time discretization of \eqref{algo1}, we adopt the third‑order strong stability preserving Runge–Kutta scheme (SSP-RK) \cite{gottlieb2001strong}. Numerical tests indicate that this solver outperforms the first‑order forward Euler method and enables avoiding unnecessary errors in free surface data. We refer the reader to the time discretization using both solvers in \cite{codepyhon}, where the time step is chosen as in \cite[Theorem 2.1]{kurganov2007second}. Specifically,
\begin{equation}\label{algo10}
\Delta t=c\frac{\Delta x}{2 \,A},
\end{equation}
where
$$A=\max_j\left\{ \max\left\{A^{+}_{j+\frac{1}{2}},-A^{-}_{j+\frac{1}{2}}\right\} \right\},\quad \text{and}\quad c\in (0,1).$$

\subsection{ Discretization of the inverse system}\label{inversedisc}
  It should be mentioned that, in practice, acquiring the measurements \eqref{surafcemeasurement} at a time $t^{\star}$ requires knowledge of the free surface elevation in a neighborhood of $t^{\star}$, or, equivalently, a short time series of observations around $t^{\star}$, to approximate the time derivatives accurately. Moreover, from the sufficient conditions $\eqref{condition on bottom}$ and \eqref{condition on bottom modified}, it can be observed that the time at which the nondegeneracy  \eqref{condition} holds lies in the second half of the well‑posedness interval $(0,t_f)$. For this reason, the time  $t^{\star}$ should be chosen close to $t_f$; for example,
$$t^{\star}=t_f-k\,\Delta t,$$
for a small $k\in \mathbb{N}$. In our numerical experiments (detailed in the next subsection), we set $t^{\star}=t_f$.

Without loss of generality, and to allow for different instants of measurements,  we extract the free surface data \eqref{surafcemeasurement} from the computed free surface elevation $\zeta$ over $(0,t_f)$  based on the forward scheme \eqref{algo1}-\eqref{algo10}. Then, we approximate the first equation in \eqref{sh4} using a second-order central finite-difference scheme in time and the explicit second-order Heun method in space (see \cite{butcher2000numerical} and the reference therein). This procedure yields numerical approximations for the discharge $q(t,x)$ and its temporal derivative $\partial_t q(t,x)$ from the measured free surface profile and the upstream discharge, see Remark \ref{abouttheinversemodel}. Next, for a given $t^{\star}$ (for example $t^{\star}=t_f$) \eqref{condition},  we use the explicit second-order Heun scheme to approximate the inverse equation in \eqref{sh4}.

\begin{enumerate}
  \item\label{step:one} Consider the following time mesh:
  $$t_0=0<t_1,t_2,\cdots,t_{N-1}<t_N=t_f,$$
  with $t^{\star}=t_{n^{\star}}$, for some index $n^{\star}\in\{0,1,....,N\}$.  At any time $t_n$, the surface’s time derivative $\partial_t\zeta(t_n,x)$ is approximated  using the following nonuniform central finite‑difference formula
  \begin{equation}\label{invalgo1}
 \partial_t\zeta(t_n,x_j)\approx\left(\partial_t\zeta\right)^n_j=\frac{-\delta_p^2\,\zeta^{n-1}_j+\left(\delta_p^2-\delta_m^2\right)\zeta^n_j+\delta^2_m\,\zeta^{n+1}_j}{\delta_m\,\delta_p\,(\delta_m+\delta_p)},
\end{equation}
where \begin{equation*}
    \begin{cases}
        \delta_p=t_{n+1}-\delta_n,\\
        \delta_m=t_n-t_{n-1}.
    \end{cases}
\end{equation*}
Here, the time step is nonuniform by construction \eqref{algo10}. Furthermore, the values $\zeta^{-1}_j$ and $\zeta^{N+1}_j$ are approximated by introducing a single ghost point just outside each boundary ($0 $ and $t_f$) and reconstructing the necessary value ($\zeta^{-1}_j$ and $\zeta^{N+1}_j$)  via quadratic Lagrange extrapolation from the three adjacent interior grid points.
  \item\label{step:two} We now proceed to solve the first equation in \eqref{sh4}: $\partial_x q=-\partial_t\zeta$.\\
   \begin{equation}\label{invalgo2}
  q_{j+1}^n=q^n_j-\frac{\Delta x}{2}\left(\left(\partial_t\zeta\right)^n_j+\left(\partial_t\zeta\right)^n_{j+1}\right).
\end{equation}
\item\label{step:three} We have $\zeta$ and $q$ on the time-space domain $(0,t_f)\times I$; we then compute the temporal derivative $\partial_t q$ and the spatial derivative $\partial_x\zeta$ by applying the second‑order finite‑difference scheme \eqref{invalgo1}. To approximate $\partial_x\zeta$, we switch the roles of time and space in \eqref{invalgo1}. Moreover, as noted in Remark \ref{abouttheinversemodel}, $\partial_t q$ can be equivalently reconstructed from the knowledge of $\partial_t^2\zeta$ and the time derivative of the inlet discharge $q_t(t,a_1)$, via the second equation in \eqref{sh4}.
\item\label{step:four}The algorithm concludes by solving the inverse equation in \eqref{sh4} using a given measurement time $t^{\star}=t_{n^{\star}}$ \eqref{condition}:
\begin{itemize}
  \item Predictor step
      \begin{equation}\label{invalgo3}
       \begin{cases}
          \Phi_1=-g\,\left(\partial_x\zeta\right)^{n^{\star}}_j\,\left(q^{n^{\star}}_j\right)^2\,\frac{1}{\phi^{n^{\star}}_j}-\left(\partial_t q\right)^{n^{\star}}_j\\
          \Phi_1^{\divideontimes}=\phi^{n^{\star}}_j+\Delta x\,\Phi_1
      \end{cases}
  \end{equation}
  
  \item Corrector step
\begin{equation}\label{invalgo4}
       \Phi_2=-g\,\left(\partial_x\zeta\right)^{n^{\star}}_{j+1}\,\left(q^{n^{\star}}_{j+1}\right)^2\,\frac{1}{\Phi_1^{\divideontimes}}-\left(\partial_t q\right)^{n^{\star}}_{j+1}
  \end{equation}
  \item Update step
  \begin{equation}\label{invalgo5}
       \phi^{n^{\star}}_{j+1}=\phi_j^{n^{\star}}+\frac{\Delta x}{2}\left(\Phi_1+\Phi_2\right).
  \end{equation}
\end{itemize}
 The step \ref{step:four} starts with a given value of upstream bed profile $b(a_1)$. Finally, the bathymetry is given by \eqref{bottomsolution}.
\end{enumerate}

From the above algorithm \ref{step:one}--\ref{step:four}, it is clear that the inverse system is solved solely using the single instant free surface measurements in \eqref{surafcemeasurement}, as well as the inlet discharge and the bed topography at the upstream of the open channel. In the following, we present numerical tests that support this approach and ensure the validity of the nondegeneracy assumption 
\eqref{condition}.

\subsection{Numerical tests}\label{testn}
In the following numerical tests, we consider a computational domain of length 25 m in the open channel (Figure \ref{fig:inverse_domain}). For each test, the bottom profile, initial data, and boundary conditions are prescribed to perform the finite-volume discretization of the forward model \eqref{algo1} -\eqref{algo10}, which produces the surface measurement vector $S_m$ \eqref{surafcemeasurement}. We then inject these surface measurements, along with the inlet discharge and the upstream bed profile, into the inverse algorithm \ref{step:one}-\ref{step:four}, resulting in an approximation of the bottom profile. Before performing the final step of this analytical approach \eqref{invalgo3}-\eqref{invalgo5}, we first verify the general assumption \eqref{condition} and its sufficient conditions \eqref{condition on bottom}-\eqref{condition on bottom modified}. To measure the accuracy of the method, we evaluate the relative $L^{\infty}$ and $L^2$ error norms, defined by
\begin{equation}\label{norms}
    \begin{cases}
        L^{\infty}_r=\frac{\bigl\|b-b_{ex}\bigr\|_{L^{\infty}(I)}}{\bigl\|b_{ex}\bigr\|_{L^{\infty}(I)}},\\
        L^{2}_r=\frac{\bigl\|b-b_{ex}\bigr\|_{L^{2}(I)}}{\bigl\|b_{ex}\bigr\|_{L^{2}(I)}},\\
    \end{cases}
\end{equation}
where $b_{ex}$ stands for the exact bottom profile and $b$ for the approximated bottom profile. For all the following numerical tests, we choose the time step $\Delta t$ according to \eqref{algo10}, where the parameter $c=0.9.$ Furthermore,  $\theta$ in \eqref{algo8} is set to $\theta=1$. \\

\noindent  {\bf \em Test 1: Steady subcritical flow over a bump} \vspace*{0.25cm}
$\quad$\\
We begin our numerical investigation with the classical steady subcritical flow over a bump (see, for example, \cite{liang2009adaptive,gessese2011reconstruction}). The bump profile is given by
\begin{equation}\label{bottombup}
   b(x)=\begin{cases}
        0.2-0.05(x-10)^2,& x\in (8,12),\\
        0,& \text{otherwise,}
    \end{cases}
\end{equation}
while the upstream discharge $q(t,0)$ and downstream depth $h(t,25)$ are fixed to :
\begin{equation}\label{boundarydatat1}
    q(t,0)=4.42 \,\text{m}^3/s \quad \text{and}\quad h(t,25)=2\,\text{m}.
\end{equation}
The initial conditions are set to:
\begin{equation*}
   q(0,x)=4.42\,\text{m}^3/s  \quad \text{and}\quad \zeta(0,x)=2\,\text{m}.
\end{equation*}
We observed numerically that this classical case falls into the strong subcritical regime \eqref{subcriticalflow} with $c_1=6$ and $c_2=1.4$. As noted in Remark \ref{problem in condition}, the left‑going wave speed $c_2$ is smaller than the right‑going wave speed $c_1$.

We discretize the $25$ m channel domain with a uniform grid $\Delta x=\frac{25}{N_x}$, where $N_x=100$, which corresponds to four surface observations per meter.  To guarantee that the flow arrives at steady state, we set the final time of simulation to $t_f=400\,s$. Figure \ref{fig:direct and inverse discharge} shows the discharge profile produced by the direct solver \eqref{algo1}-\eqref{algo10} alongside the discharge reconstructed using the measured free surface elevation and the inlet discharge \eqref{invalgo2}.
\begin{figure}[!htb]
  \centering
  \includegraphics[width=0.7\textwidth]{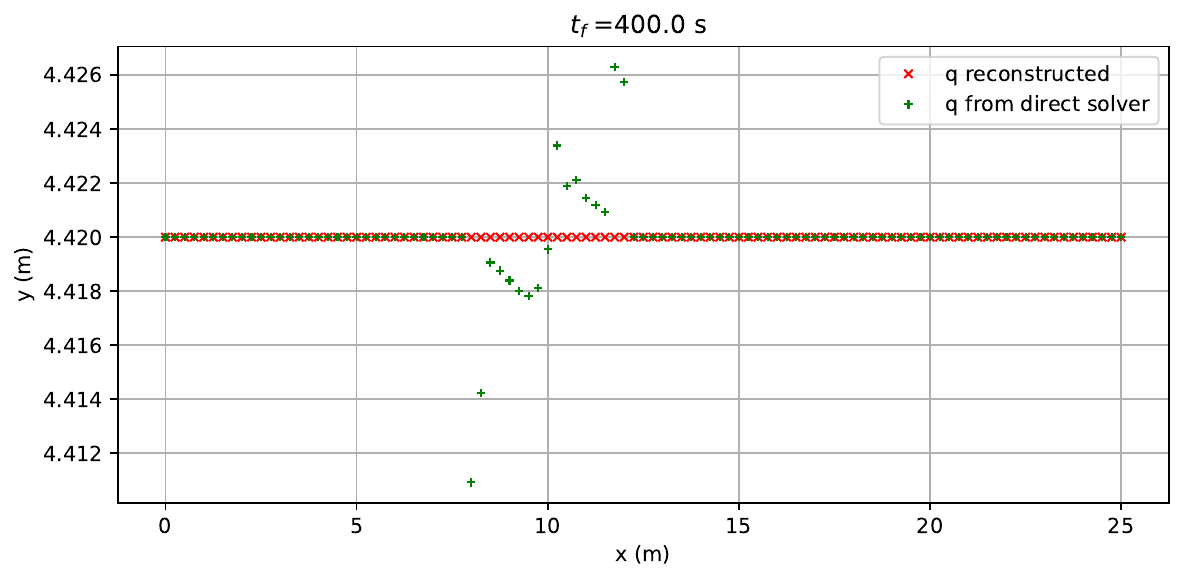}
  \caption{ Comparison of the forward and reconstructed discharge profiles.}
  \label{fig:direct and inverse discharge}
\end{figure}
This Figure \ref{fig:direct and inverse discharge} demonstrates that both the forward‑computed and the reconstructed discharges satisfy the steady state. Moreover, the reconstructed has a significantly smaller deviation from the steady state value $q=4.42$ than the simulated. However, increasing $N_x$ further improves the simulated discharge convergence toward the steady-state discharge. Finally, it should be mentioned that this is a simple case for which the required nondegeneracy  \eqref{condition} holds automatically.

Since we are in the steady state flow case, from \eqref{sh5}-\eqref{simpliciation for discharge} we have an analytic solution for the reconstructed bed profile. In Figure \ref{fig:surface analy bottom}, we plot the actual bathymetry, the analytic solution, and the free surface at time $t_f=400$.
\begin{figure}[!htb]
  \centering
  \begin{minipage}{0.48\textwidth}
    \includegraphics[width=\linewidth]{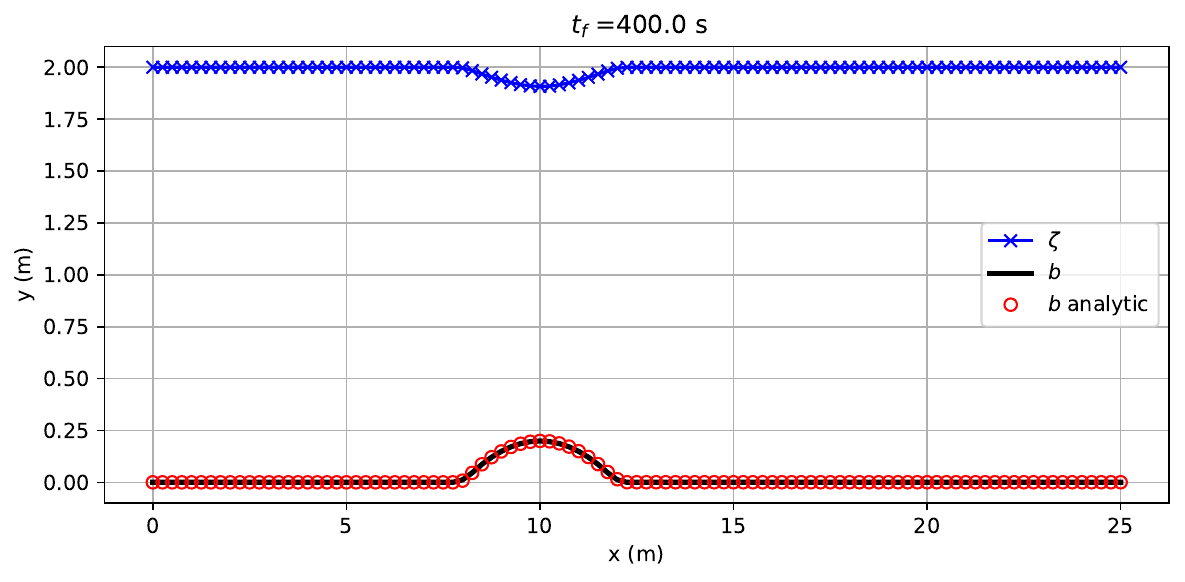}
  \end{minipage}\hfill
  \begin{minipage}{0.48\textwidth}
    \includegraphics[width=\linewidth]{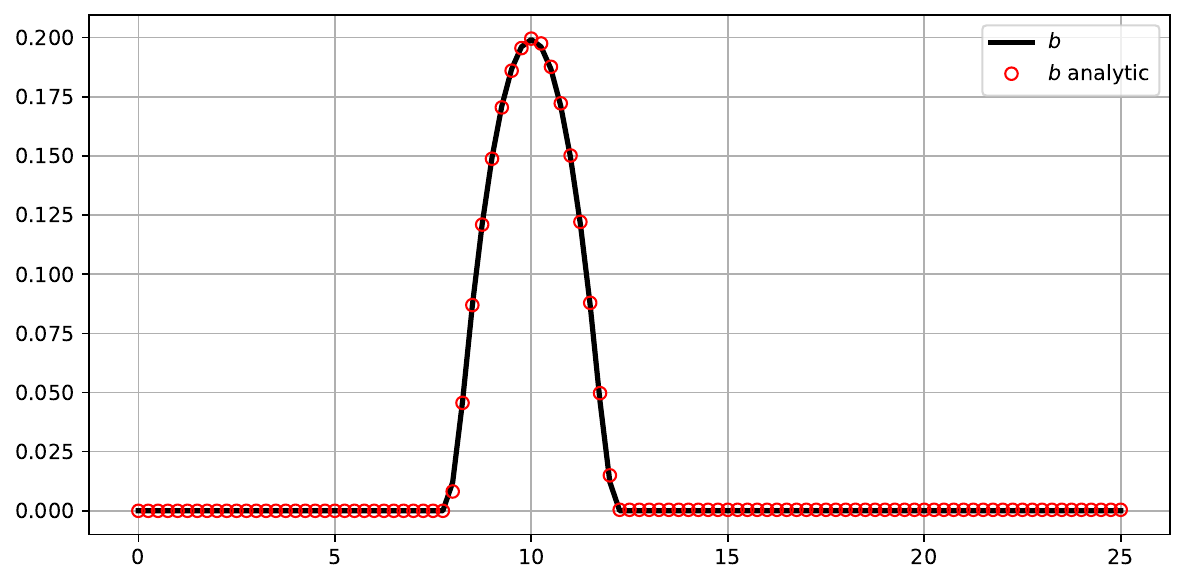}
  \end{minipage}

  \caption{Comparison of the analytic bed topography and true bathymetry.}
  \label{fig:surface analy bottom}
\end{figure}
The relative $L^{\infty}$ and $L^2$ errors defined in \eqref{norms} are very small. Precisely, the relative $L^{\infty}$  error is 1.98 \%, and the relative $L^2$ error is 1.28 \%. Moreover, these errors are only introduced by the forward solver \eqref{algo1}-\eqref{algo10}. Furthermore, when we increase $N_x$, these small errors vanish, yielding an exact reconstruction of the bed profile.  This, in particular, confirms the accuracy of the forward discretization scheme.

For the sake of comparison with \cite{gessese2011reconstruction,gessese2011inferring} and although we have an analytical solution, we apply our general inverse approach \eqref{invalgo1}-\eqref{invalgo5} to this particular case. Figure \ref{fig:surface recon bottom test1} shows that the reconstructed bathymetry is in good agreement with the true channel bed profile. Furthermore, the relative errors \eqref{norms} are small: the relative $L^{\infty}$ error is 5.26 \% and the relative $L^2$ error is 2.26 \%. As in the analytic case, these errors diminish sharply as $N_x$ increases. For example, at $N_x=200$ (which still represents a realistic number of observations), the relative $L^{\infty}$ error drops to 2.58\% and the relative $L^2$ error to 0.98 \%.

\begin{figure}[!htb]
  \centering
  \begin{minipage}{0.48\textwidth}
    \includegraphics[width=\linewidth]{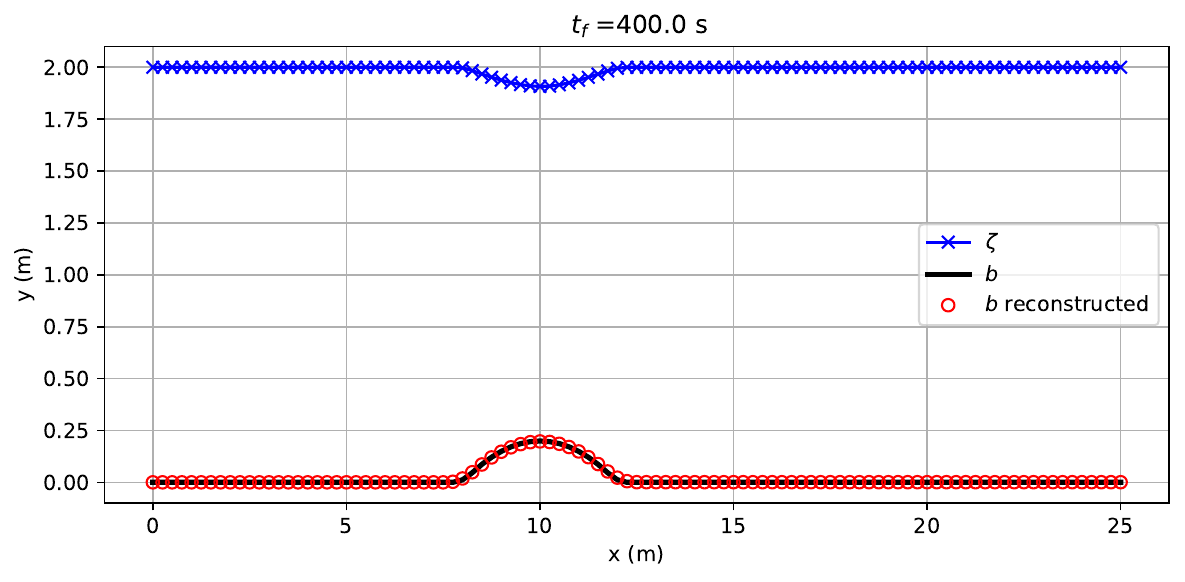}
  \end{minipage}\hfill
  \begin{minipage}{0.48\textwidth}
    \includegraphics[width=\linewidth]{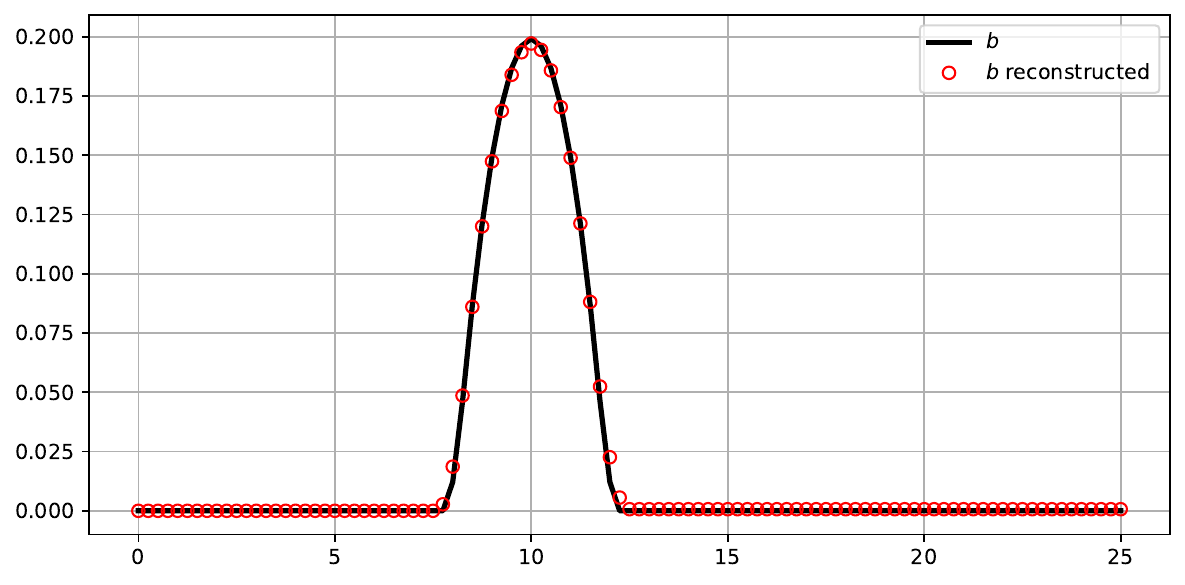}
  \end{minipage}

  \caption{Comparison of the reconstructed bed topography and true bathymetry.}
  \label{fig:surface recon bottom test1}
\end{figure}

Concluding the first test, we emphasize that our method does not depend on the flow having reached the steady state. Even for smaller $t_f$ such as $t_f\geq1\,s$, we obtain similar results and observations compared to the steady state $t_f=400\, s$. Furthermore, it should be noted that, in all our numerical tests for this case, the sufficient conditions \eqref{condition on bottom} and \eqref{condition on bottom modified} are satisfied over the time intervals $(0,t_f)$. For brevity, the numerical details are omitted; for any further numerical tests, we refer the reader to the Python script of this paper  \cite{codepyhon}.\\

\noindent  {\bf \em Test 2:  Unsteady subcritical flow over a bump and  impact of initial data} \vspace*{0.25cm}
$\quad$\\
In this test, we evaluate the adaptability of our method (see Remark \eqref{r1}) by examining a scenario in which the hypothesis \eqref{subcriticalflow} fails locally in time. Moreover,  we illustrate the impact of the chosen initial data on the flow features. 

To enable a direct comparison with the previous test, we retain the bottom topography given by \eqref{bottombup}  and impose the same boundary conditions \eqref{boundarydatat1}. Instead, we initialize the flow with the stationary initial conditions:
\begin{equation}\label{initila2}
    q(0,x)=0\,\text{m}^3/s,  \quad \text{and}\quad \zeta(0,x)=2\,\text{m}
\end{equation}
This scenario reflects an unsteady inflow from upstream meets the stationary flow over $(0,25)$. For the numerical implementation of this case, we use the same spatial step $\Delta x=\frac{25}{100}$; however, we only simulate up to $t_f=50\,s$, well before the steady state is reached, see Figure \ref{fig:direct and inverse discharge1}.

\begin{figure}[!htb]
  \centering
  \includegraphics[width=0.7\textwidth]{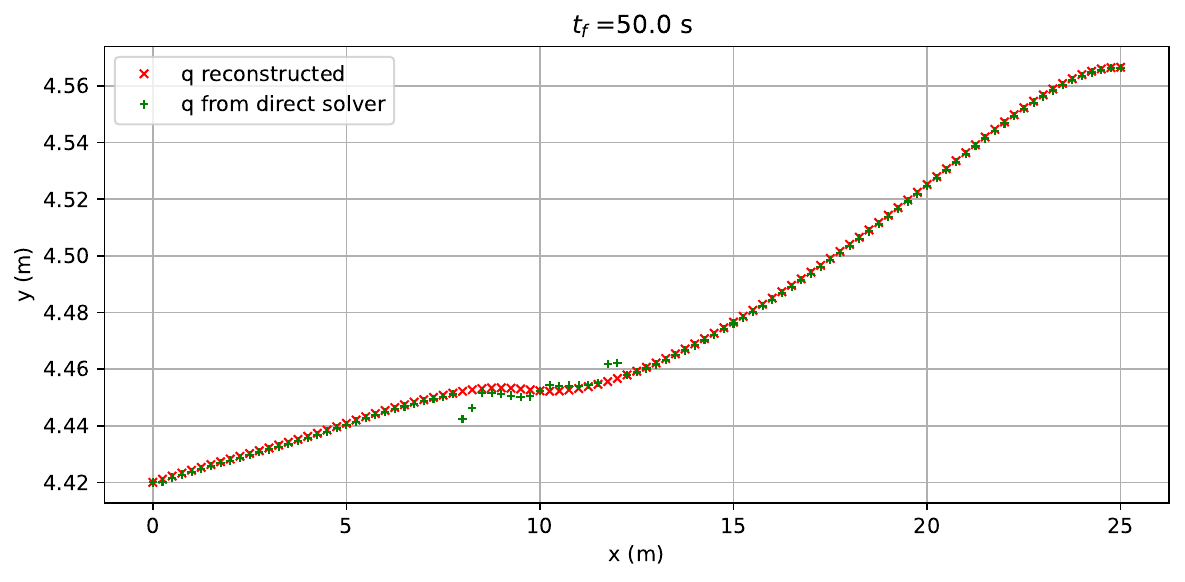}
  \caption{ Comparison of the forward and reconstructed discharge profiles.}
  \label{fig:direct and inverse discharge1}
\end{figure}

In the previous test, the strong subcriticality hypothesis \eqref{subcriticalflow} is verified on $(0,t_f)$ for any $t_f\geq1\,s$. However, due to the modified initial conditions \eqref{initila2}, this hypothesis is invalidated around $t=20\,s$. As illustrated in Figure \ref{fig: lambda 2}, the left‐going characteristic speed $\lambda_2$ crosses zero twice, thereby breaking hypothesis \eqref{subcriticalflow} and thus invalidating the corresponding sufficient condition \eqref{condition on bottom}. Based on  Remark \ref{r1}, we only require \eqref{subcriticalflow} to hold from a shifted time $t_s$. Hence, if we choose $t_s=27.39\,s$, hypothesis \eqref{subcriticalflow} is satisfied for any $t\in(27.39,t_f)$ with $c_1=6.3$ and $c_2=1.15$. Furthermore, $t_s+25/c_0=49.12<t_f$. Note that we can choose any $t_s>27$ and take $t_f$ large enough. Using this shift, both hypotheses \eqref{subcriticalflow} and \eqref{transcriticallflow}, as well as the associated sufficient conditions \eqref{condition on bottom} and \eqref{condition on bottom modified}, are satisfied, which guarantees the nondegeneracy \eqref{donesubcretical} and \eqref{donesubcreticalmod} for any $t\in(49.12,50)$.

\begin{figure}[!htb]
  \centering
  \includegraphics[width=0.7\textwidth]{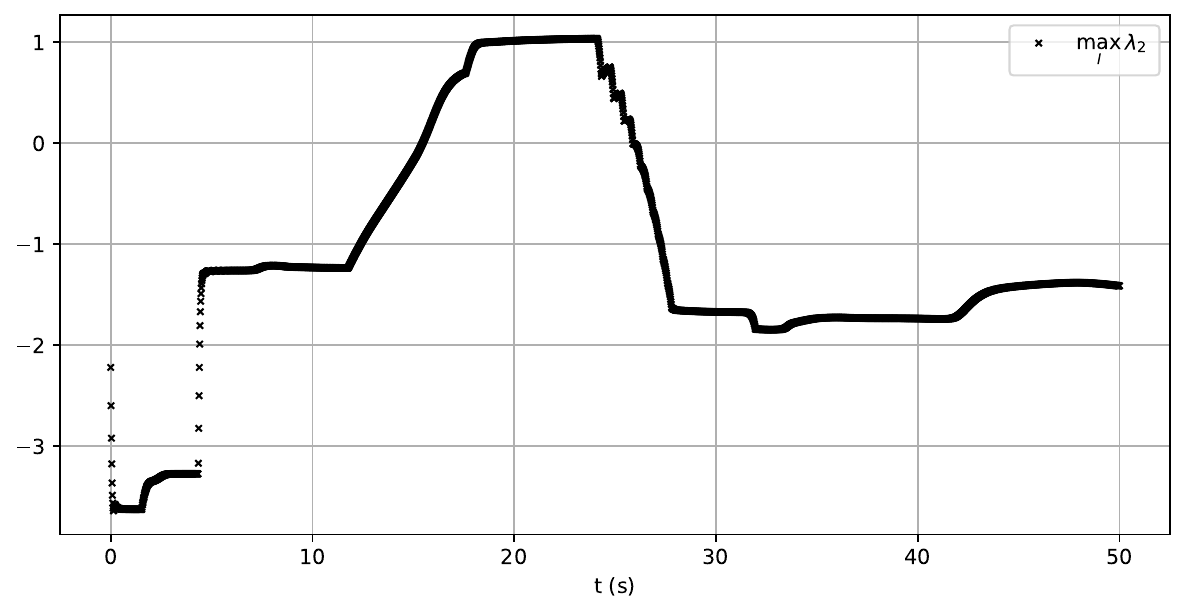}
  \caption{ The maximum in space of the left‐going characteristic speed.}
  \label{fig: lambda 2}
\end{figure}

We now compare the reconstructed and exact bathymetries. Figure \ref{fig:surface recon bottom test2} shows that the bed profile is accurately detected. The relative $L^{\infty}$ error is 5.22\% and the relative $L^2$ error is 3.35\%. These errors decrease rapidly as $N_x$ increases; for example, if we choose $N_x=200$, the relative $L^{\infty}$  error is 2.51\% and the relative $L^2$ error is 1.13\%. 

\begin{figure}[!htb]
  \centering
  \begin{minipage}{0.48\textwidth}
    \includegraphics[width=\linewidth]{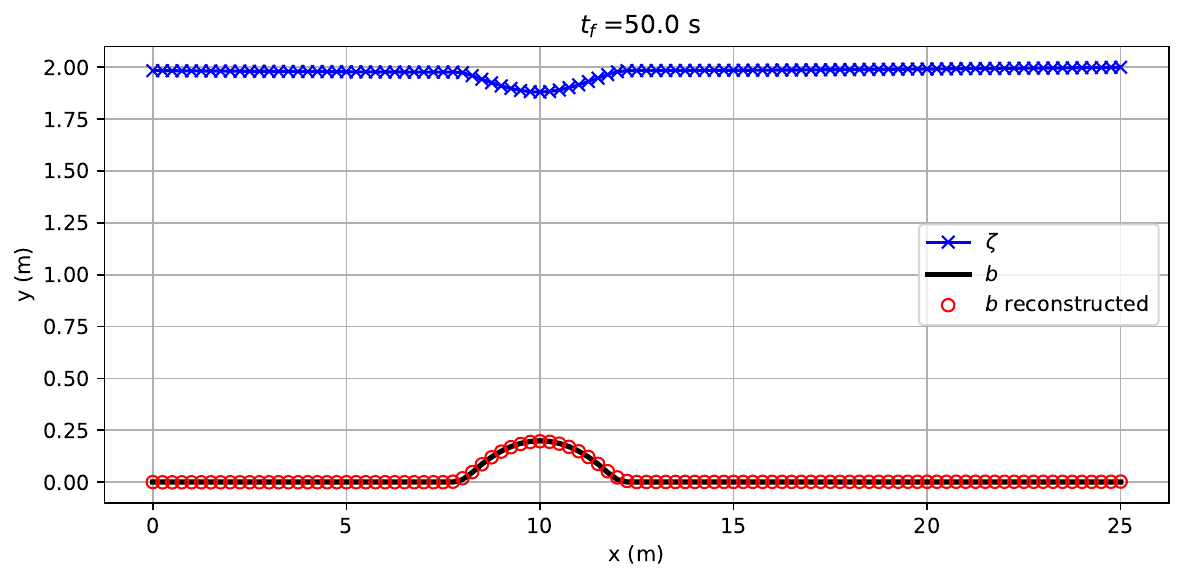}
  \end{minipage}\hfill
  \begin{minipage}{0.48\textwidth}
    \includegraphics[width=\linewidth]{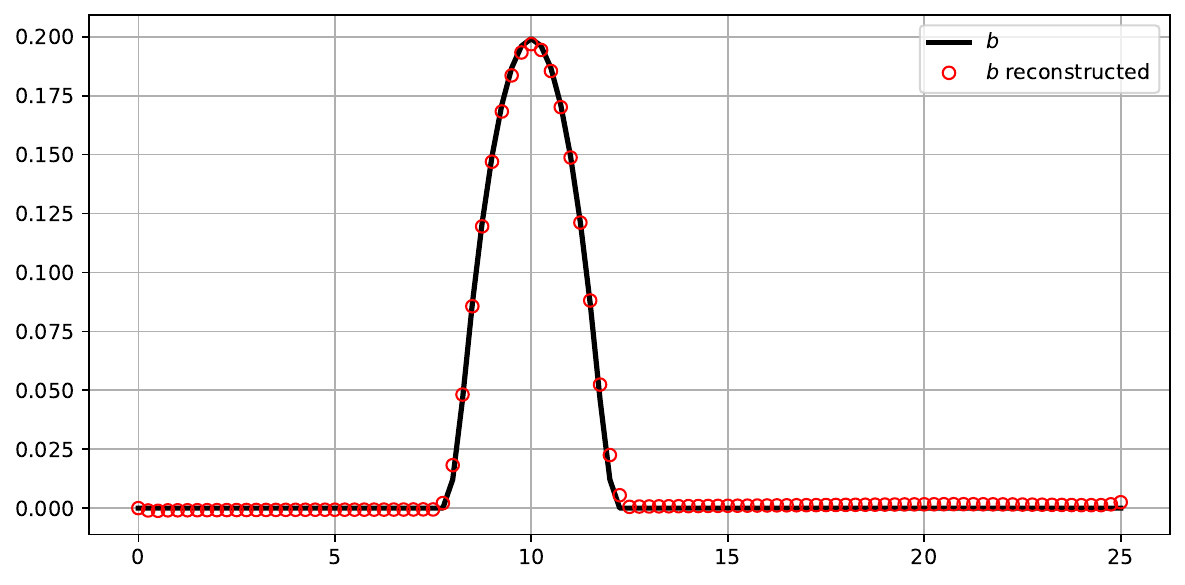}
  \end{minipage}

  \caption{Comparison of the reconstructed bed topography and true bathymetry.}
  \label{fig:surface recon bottom test2}
\end{figure}

In both tests, we showed that the sufficient conditions \eqref{condition on bottom} and \eqref{condition on bottom modified} remain valid and adaptable when considering subcritical flows. Moreover, we demonstrated that the approach provides high reconstruction accuracy. Before extending our study to other types of flow, it should be mentioned that the results of the above two tests also apply to different bed shapes, such as the Gaussian bump examined in \cite{nicholls2009detection,khan2021variational}. For brevity, we do not present these supplementary simulations here; interested readers can produce the results using \cite{codepyhon}.\\

\noindent  {\bf \em Test 3: Unsteady supercritical flow over a sandbar } \vspace*{0.25cm}
$\quad$\\
Assuming a supercritical flow ($\mathrm{Fr}>1$) automatically implies the nondegeneracy assumption \eqref{condition}, eliminating the need for supplementary sufficient conditions  \eqref{condition on bottom} and \eqref{condition on bottom modified}. The Froude number ($\mathrm{Fr}$) is defined by 
\begin{equation}\label{Fr}
    \mathrm{Fr}(t,x)=\frac{u(t,x)}{\sqrt{g\,h(t,x)}},\quad \forall (t,x)\in (0,t_f)\times I.
\end{equation}
In this test, we consider global supercritical flow. Precisely
\begin{equation}\label{asumption froude}
    \mathrm{Fr}(t,x)>1,\quad \forall (t,x)\in (0,t_f)\times I.
\end{equation}
However, similar analysis can be applied to the strong supercritical flow \cite{huang2011one} and the strengthened supercritical flow \cite{kounadis2020galerkin}, since all these flow hypotheses lead to $u(t,x)>0$ for every $t\in (0,t_f)$ and $x\in I$. In order to produce the surface data required for our analytical approach, we adopt the sandbar bed profile  \cite{nicholls2009detection,khan2021variational}. For a given $A\in\mathbb{R}_{+}^{\star}$, the sandbar bathymetry is defined by
\begin{equation}\label{sandbar}
   b(x)= A\left(\tanh\left(2\left(\frac{6\,x}{25}-3+B\right)\right)- \tanh\left(2\left(\frac{6\,x}{25}-3-B\right)\right)   \right),\quad \forall x\in I,
\end{equation}
where $B=\frac{3\pi}{5}$. At the upstream boundary, we prescribe a constant inlet depth  and a time‑periodic inlet discharge:
\begin{equation*}
    q(t,0)=2.5+0.3\sin\left(\frac{2\pi\,t}{10}\right) \,\text{m}^3/s \quad \text{and}\quad h(t,0)=0.6\,\text{m}.
\end{equation*}
The last step before executing the forward solver is to define the initial state. Here, we consider the following non‑stationary initial conditions
\begin{equation*}
     q(0,x)=2.5\,\text{m}^3/s,  \quad \text{and}\quad \zeta(0,x)=0.6+0.15\,e^{-\left(\frac{x-5}{1.2}\right)^2}\,\text{m}.
\end{equation*}
In the following figures of this test, we fix the sandbar parameter \eqref{sandbar} to $a=0.1$. Numerical tests indicate that the flow remains supercritical and meets condition \eqref{supercritical} until $t_f=12\,s $; beyond $t_f=13\,s$, global supercriticality \eqref{supercritical} is lost.  However, the nondegeneracy assumption \eqref{condition} holds for any $t_f$. Without loss of generality, we chose $t_f=12\,s$. For the spatial discretization, we keep the same grid step $\Delta x=\frac{25}{100}$. 

Figure \ref{fig: Fr} presents the Froude number distribution along the channel at $t_f$, clearly demonstrating supercritical flow at $t_f$. Although the flow is globally supercritical up to $12\,s$, we recall that our approach only requires a single time for which strict positivity \eqref{condition} holds. Hence, we only focus on $t_f$.

\begin{figure}[!htb]
  \centering
  \includegraphics[width=0.7\textwidth]{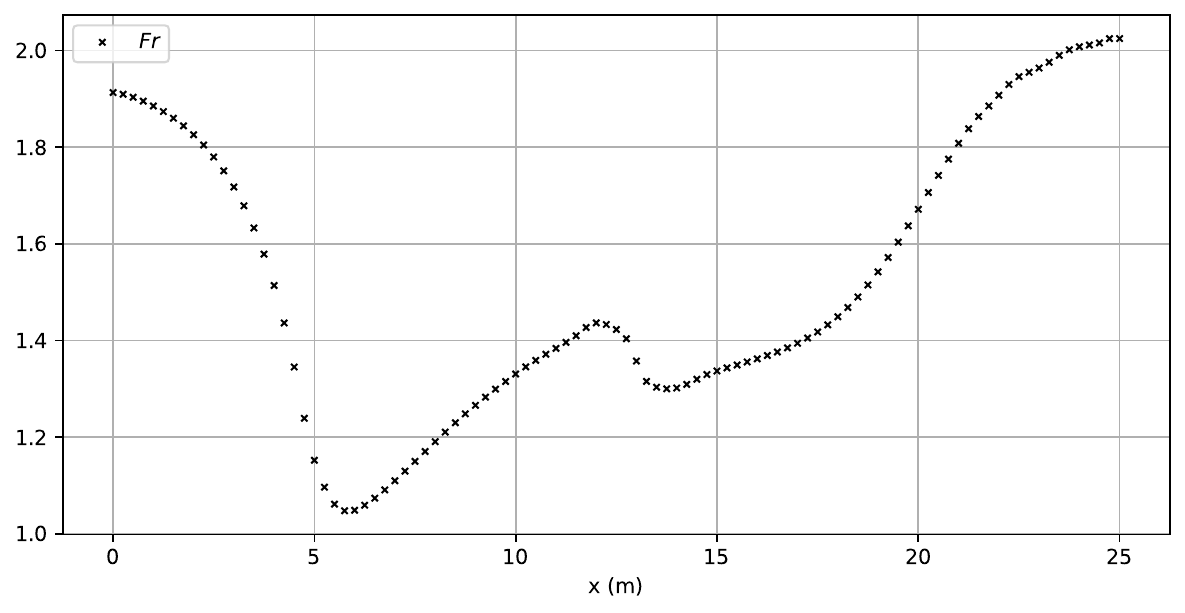}
  \caption{  Spatial profile of the Froude number \eqref{Fr} at $t_f=12\,s$.}
  \label{fig: Fr}
\end{figure}

Using the generated surface measurements \eqref{surafcemeasurement} and the inlet discharge, we reconstruct the spatial discharge profile  \eqref{invalgo2}. Figure \ref{fig:direct and inverse discharge2} presents the forward‑computed discharge and the discharge reconstructed from surface measurements at $t_f=12\,s$.

\begin{figure}[!htb]
  \centering
  \includegraphics[width=0.7\textwidth]{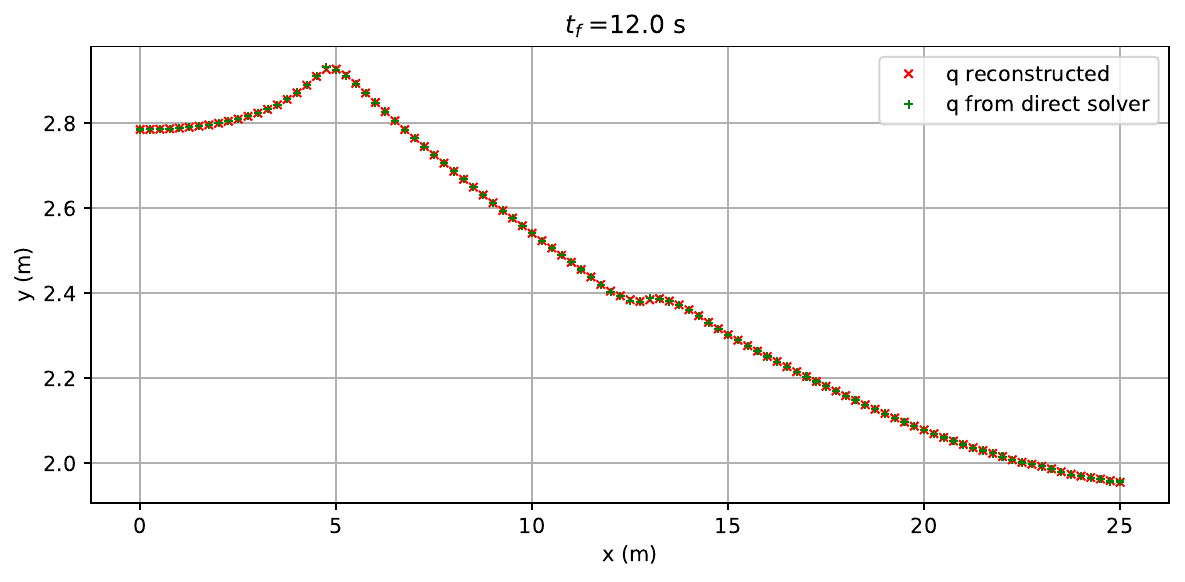}
  \caption{ Comparison of the forward and reconstructed discharge profiles.}
  \label{fig:direct and inverse discharge2}
\end{figure}

We now illustrate the accuracy of our approach to reconstruct the bottom profile. In Figure \ref{fig:surface recon bottom test3}, it can be seen that the bed profile is very well recovered. Moreover, the relative $L^{\infty}$ error is 1.28\% and the relative $L^2$ error is 0.64\%. Importantly, for a grid resolution of $N_x=100$, the error of estimation achieved reflects excellent reconstruction performance.

\begin{figure}[!htb]
  \centering
  \begin{minipage}{0.48\textwidth}
    \includegraphics[width=\linewidth]{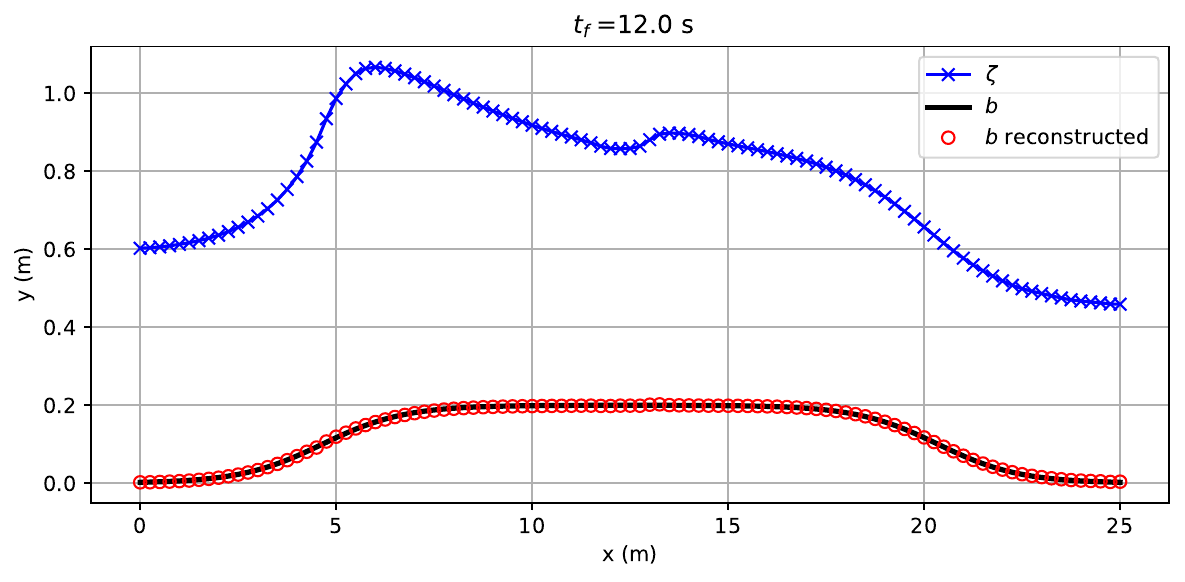}
  \end{minipage}\hfill
  \begin{minipage}{0.48\textwidth}
    \includegraphics[width=\linewidth]{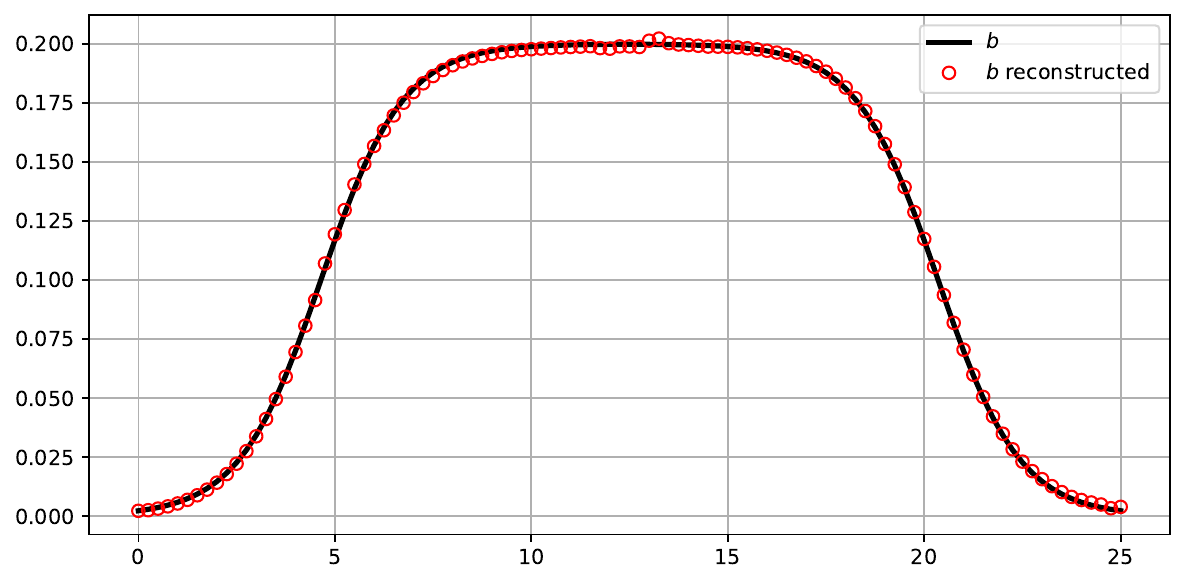}
  \end{minipage}

  \caption{Comparison of the reconstructed bed topography and true bathymetry.}
  \label{fig:surface recon bottom test3}
\end{figure}

At the end of this test, it should be mentioned that the reconstruction accuracy remains consistent across all time instants $t\in(0,12)$ and even for $t\geq 13$ when hypothesis \eqref{asumption froude} is no longer valid. Furthermore, an equivalent accuracy is achieved for different values of the sandbar amplitude $A$ \eqref{sandbar}, which underlines the flexibility of our method.

\noindent  {\bf \em Test 4:  Hydraulic fall of unsteady transcritical flow over a Gaussian topography} \vspace*{0.25cm}
$\quad$\\
In this test, we examine an unsteady hydraulic fall scenario \cite{tam2015predicting}. Consistent with our earlier tests, the bottom topography, inlet/outlet conditions, and initial state are specified to produce the single time measurements \eqref{surafcemeasurement}. We recall that our bathymetry‑reconstruction approach does not rely on downstream observations or particular initial data; it only requires inlet discharge, upstream bottom value, and the free surface measurements given in \eqref{surafcemeasurement}.  We simulate the hydraulic fall over the following Gaussian bottom profile \cite{nicholls2009detection}:
\begin{equation*}
    b(x)=0.2\,\sech\left(2\left(\frac{6\,x}{25}-3\right)\right),\quad\forall x\in I.
\end{equation*}
We initialize the flow with
\begin{equation*}
     q(0,x)=1.5\,\text{m}^3/s,  \quad \text{and}\quad \zeta(0,x)=0.7\,\text{m}.
\end{equation*}
In order to avoid any steady‑state constraint, the inlet discharge is specified as follows
\begin{equation*}
     q(t,0)=1.5+0.2\sin\left(\frac{2\pi\,t}{10}\right)\,\text{m}^3/s.
\end{equation*}

With $\Delta x=\frac{25}{100}$, we run the forward solver \eqref{algo1}-\eqref{algo10} up to $t_f=20\,s$. From the obtained surface measurements \eqref{surafcemeasurement} and the given inlet flow rate, we reconstruct the discharge profile over the domain. As shown in  Figure \ref{fig:direct and inverse discharge3}, the recovered discharge aligns with the direct‑solver output discharge. In Figure \ref{fig: Fr2}, we plot the Froude number at $t_f$, revealing a single crossing of unity and thus confirming that the flow is transcritical.

\begin{figure}[!htb]
  \centering
  \includegraphics[width=0.7\textwidth]{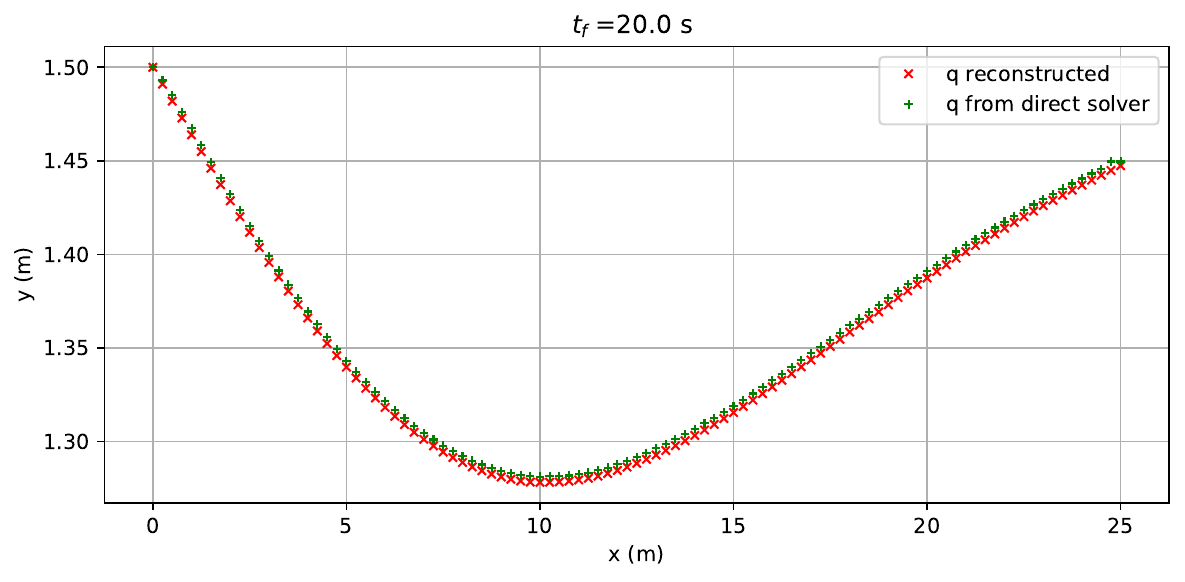}
  \caption{ Comparison of the forward and reconstructed discharge profiles.}
  \label{fig:direct and inverse discharge3}
\end{figure}

\begin{figure}[!htb]
  \centering
  \includegraphics[width=0.7\textwidth]{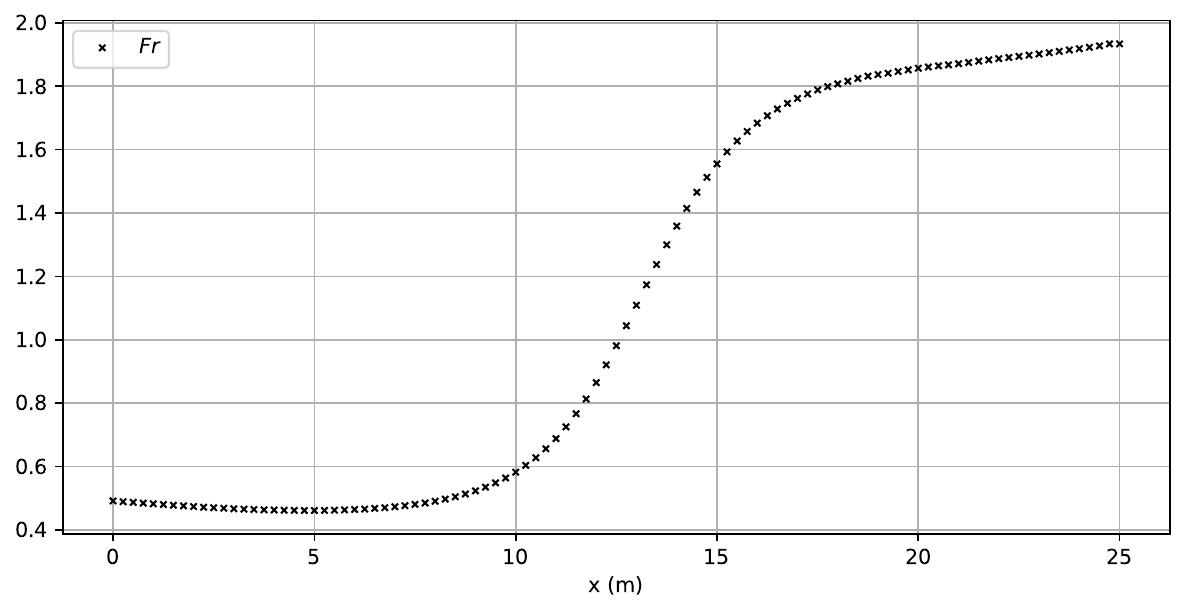}
  \caption{  Spatial profile of the Froude number \eqref{Fr} at $t_f=20\,s$.}
  \label{fig: Fr2}
\end{figure}
Applying the algorithm \eqref{invalgo1}–\eqref{invalgo5} produces an estimated bottom profile whose shape matches the true bathymetry, as shown in Figure \ref{fig:surface recon bottom test4}. Quantitatively, we measure a relative $L^{\infty}$ error of 1.88 \% and a relative $L^2$ error of 2.46 \%. Spatial refinement further reduces these small errors: with $N_x=200$, the relative $L^{\infty}$error drops to 0.93 \% and the relative $L^2 $ error to 1.21 \%, reflecting the robustness and precision of the method.
\begin{figure}[!htb]
  \centering
  \begin{minipage}{0.48\textwidth}
    \includegraphics[width=\linewidth]{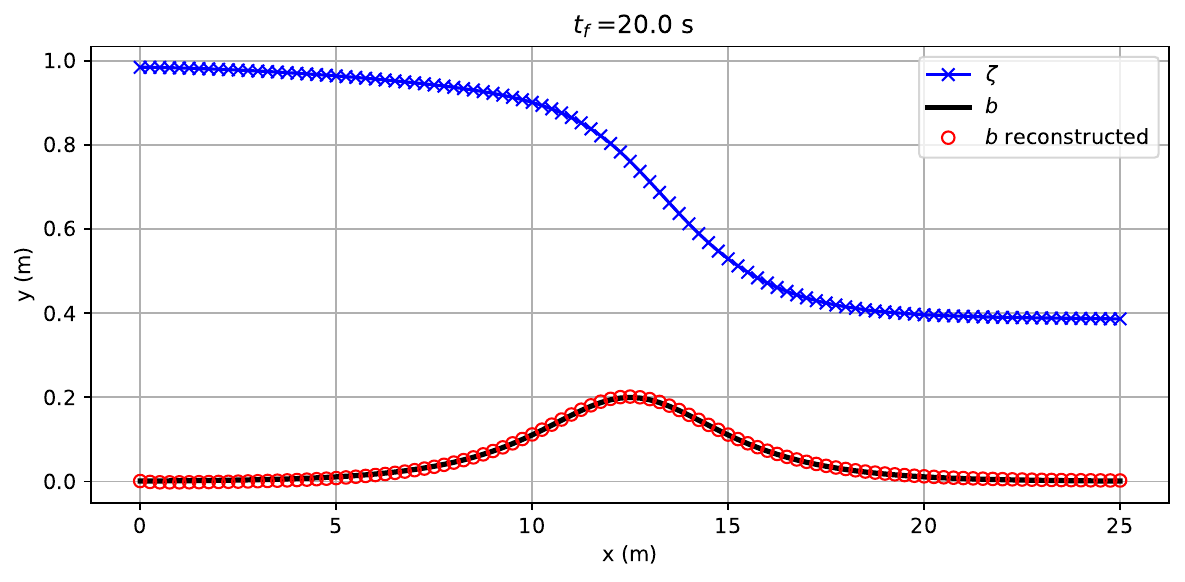}
  \end{minipage}\hfill
  \begin{minipage}{0.48\textwidth}
    \includegraphics[width=\linewidth]{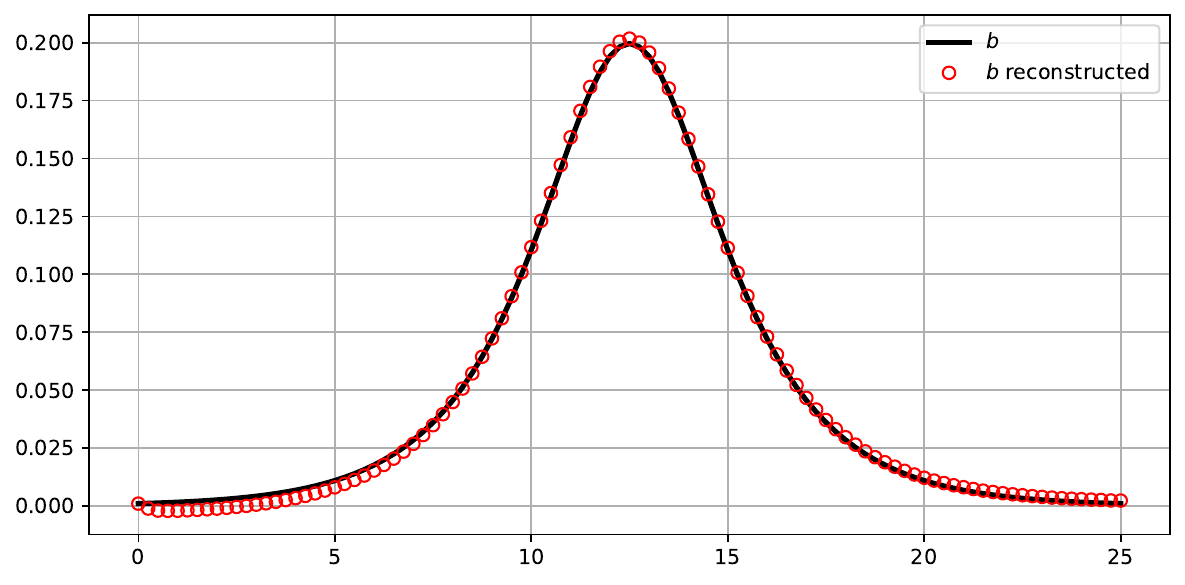}
  \end{minipage}

  \caption{Comparison of the reconstructed bed topography and true bathymetry.}
  \label{fig:surface recon bottom test4}
\end{figure}

Regarding Corollary \ref{transflow}, although Figure \ref{fig:direct and inverse discharge3} already indicates a strictly positive discharge, we carry out a numerical verification of the sufficient condition \eqref{condition on bottom modified} in order to rigorously confirm nondegeneracy. Our computation shows that this sufficient condition is met for $\rho =0.2$.\\

\noindent  {\bf \em Test 5: Impact of noisy free surface on stability} \vspace*{0.25cm}
$\quad$\\
In practice,  free surface data are noisy; therefore, analyzing the response of our approach to such noise is important for the validation. Although Proposition \ref{propositionstability} ensures stability in theory, the associated exponential factor grows with the Lipschitz constant $\bigl\|\partial_x \zeta\bigr\|_{L^{\infty}(I)}$, which tends to be large for noisy data, see Remark \ref{stability problm}. Consequently, noisy free surface measurements yield larger Lipschitz constants and hence larger stability bounds. Here, we show that applying a cubic spline filter \cite{reinsch1967smoothing}  to smooth the noisy free surface effectively controls this Lipschitz constant \eqref{C11} and enables accurate bathymetry reconstruction.

Unlike the previous examples, this test uses a non-standard bed geometry consisting of two bumps \eqref{bottombup} and a cosine function. To generate the surface data \eqref{surafcemeasurement}, we use the stationary initial conditions given in \eqref{initila2} and the downstream depth \eqref{boundarydatat1}. For the inlet discharge, we introduce a time‐dependent perturbation to the upstream discharge in  \eqref{boundarydatat1}. Precisely
$$q(t,0)=4.42+0.2\sin\left(\frac{2\pi \,t}{10}\right)\;\,\text{m}^3/s .$$

Figure \ref{fig:surface recon bottom test5} compares the reconstructed bathymetry to the actual profile for $t_f=40\,s$ and $N_x=100$. The bed profile is recovered with a relative $L^{\infty}$ error of 5.78 \% and a relative $L^2$ error of 3.06 \%, demonstrating the good accuracy of the approach to recover a wide variety of bed geometries.

\begin{figure}[!htb]
  \centering
  \begin{minipage}{0.48\textwidth}
    \includegraphics[width=\linewidth]{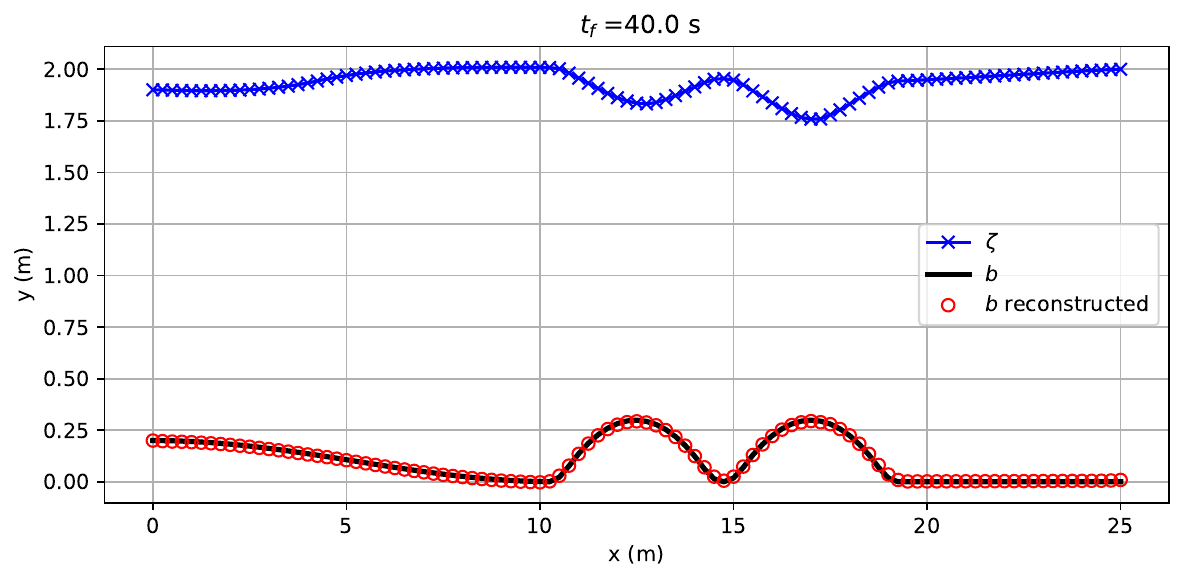}
  \end{minipage}\hfill
  \begin{minipage}{0.48\textwidth}
    \includegraphics[width=\linewidth]{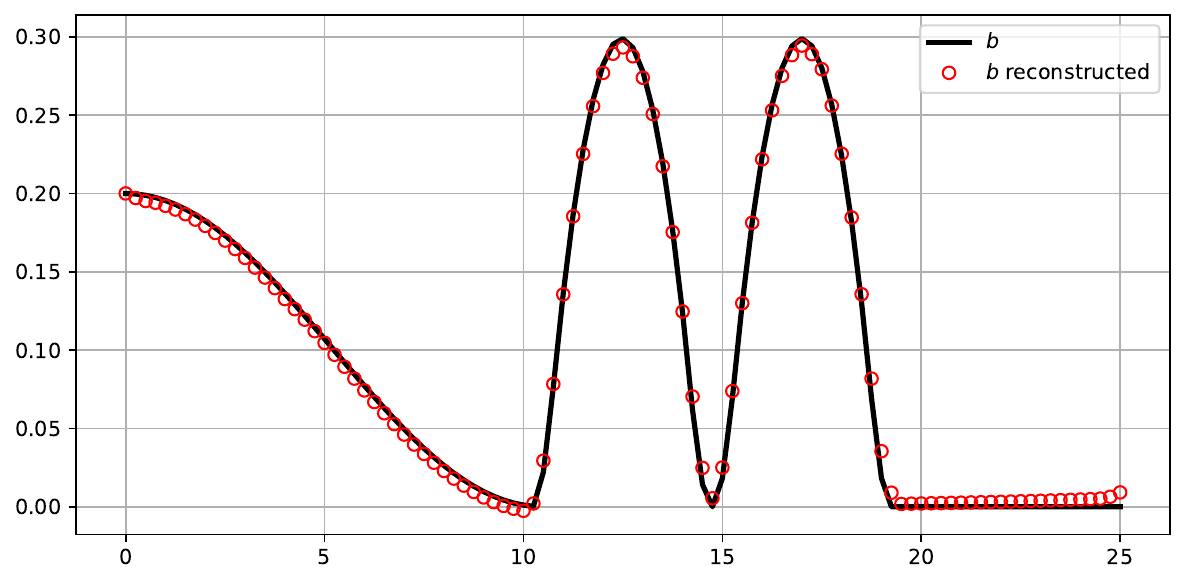}
  \end{minipage}

  \caption{Comparison of the reconstructed bed topography and true bathymetry.}
  \label{fig:surface recon bottom test5}
\end{figure}

Prior to adding noise, it should be mentioned that this test closely mirrors {\bf \em Test 2}. In particular, we find that the strong subcriticality hypothesis  \eqref{subcriticalflow} becomes valid only after a time shift $t_s$. However, unlike in {\bf \em Test 2}, beyond this instant, only the second sufficient condition \eqref{condition on bottom modified} holds. This is because, in this test, the downstream characteristic speed $c_2$ is small even after shift (equal to $0.2$) and cannot control the bottom profile variation in \eqref{condition on bottom} (see Remark \ref{problem in condition}).

In the following, we introduce 2 \% of noise based on the depth of the flow to the free surface. Figure \ref{fig:noise1} illustrates the noisy free surface alongside the exact profile. Injecting this noisy surface directly into \eqref{invalgo1}-\eqref{invalgo5} results in large approximation errors, as highlighted in Remark \ref{stability problm}. Therefore, we first smooth the noisy free surface using a cubic‑spline smoother \cite{reinsch1967smoothing}, see Figure \ref{fig:noise1}. This smoothed surface is then used in \eqref{invalgo1}-\eqref{invalgo5}.  In Figure \ref{noise2}, the exact bottom elevation and the reconstructed profile from the smoothed data are compared, demonstrating that our method recovers the bathymetry despite the initial noise. Moreover, even with the 2 \% of noise, reconstruction accuracy remains high: the relative $L^{\infty}$ error increases to 10 \% and the relative $L^2$ error to 8 \%, corresponding only to 5 \% decay compared to the noise‑free case. Although the error magnitudes depend implicitly on $N_x$, the smoothing parameter, and noise shape, the bed topography is generally well reconstructed.

\begin{figure}[!htb]
  \centering
  \begin{minipage}{0.48\textwidth}
    \includegraphics[width=\linewidth]{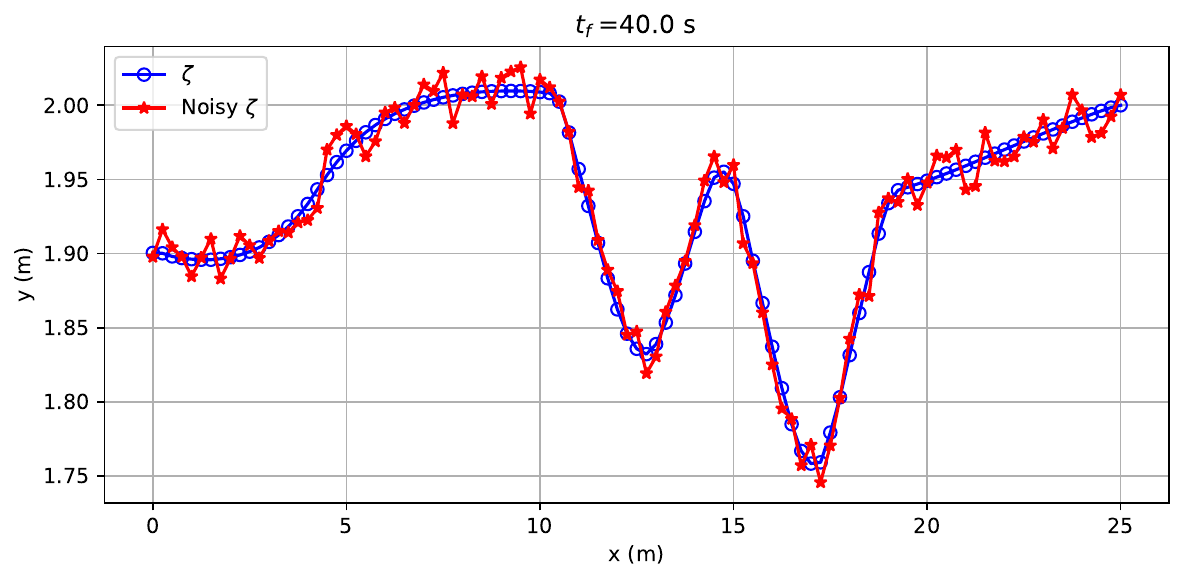}
  \end{minipage}\hfill
  \begin{minipage}{0.48\textwidth}
    \includegraphics[width=\linewidth]{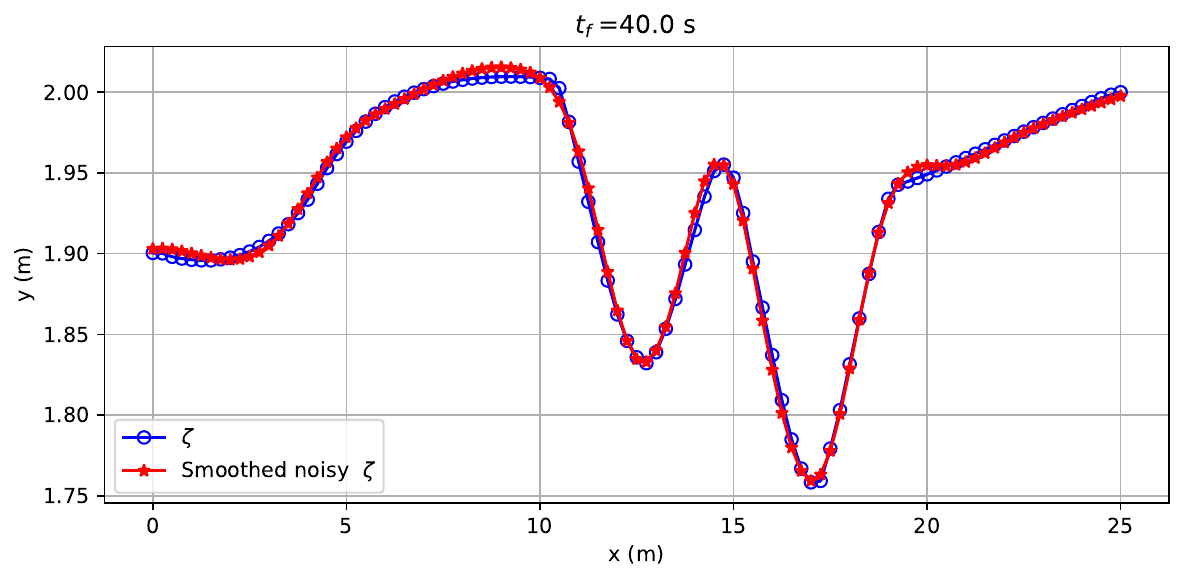}
  \end{minipage}

  \caption{ Exact, noisy, and cubic‑Spline–smoothed free surface profiles}
  \label{fig:noise1}
\end{figure}

\begin{figure}[!htb]
  \centering
  \begin{minipage}{0.48\textwidth}
    \includegraphics[width=\linewidth]{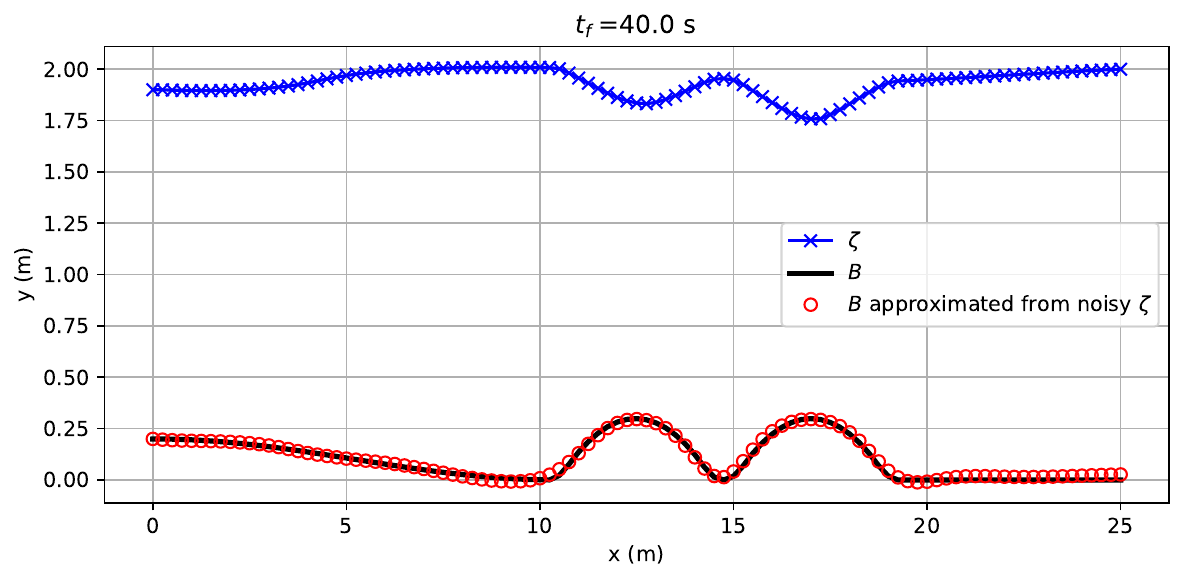}
  \end{minipage}\hfill
  \begin{minipage}{0.48\textwidth}
    \includegraphics[width=\linewidth]{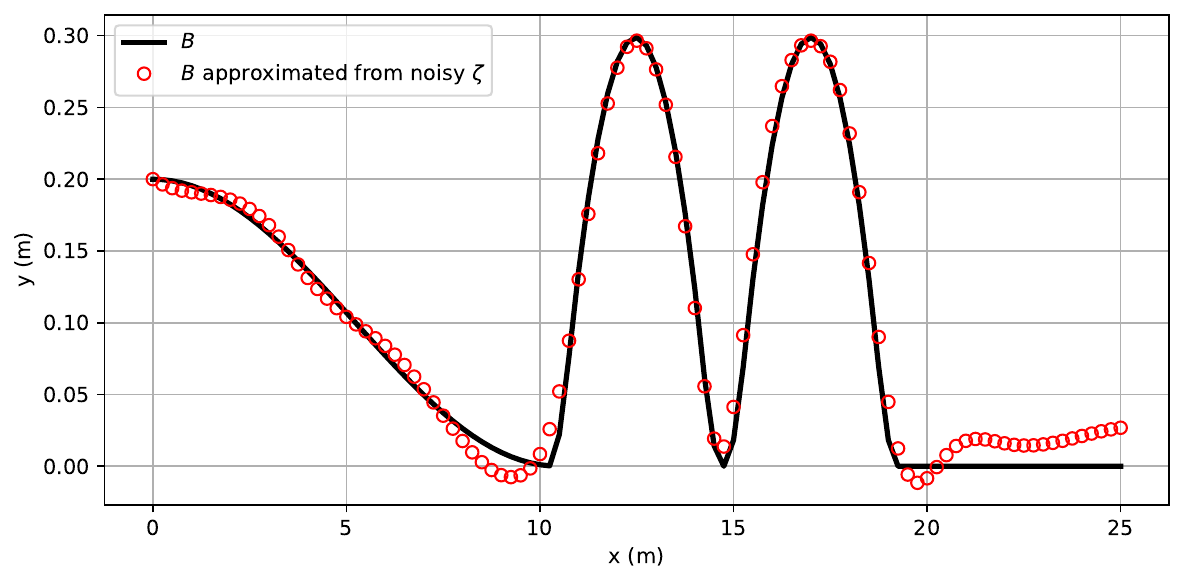}
  \end{minipage}

  \caption{  Comparison of the true and reconstructed bathymetry from noisy free surface}
  \label{noise2}
\end{figure}

At the end of this subsection, we emphasize that our reconstruction method delivers comparable performance under different settings, such as the hypothetical hydrograph inlet discharge considered in \cite{sanders2000adjoint}, as well as for more complex bottom profiles. For any further numerical simulations, we encourage readers to utilize \cite{codepyhon}.

\section{Conclusions, further comments, and future work}\label{sec:conclusions}
We have provided a direct, stable, and analytic approach to infer the bed topography using the one‑dimensional shallow water equations in both steady and unsteady flows. The technique requires an open domain setting and knowledge of the free surface elevation and its first two time derivatives over the spatial domain at a single time $t^{\star}$ \eqref{surafcemeasurement}, as well as the flow rate and the bed profile at the upstream boundary, see Figure \ref{fig:inverse_domain}. This time  $t^{\star}$ should satisfy the nondegeneracy assumption \eqref{condition}. We illustrated that this condition automatically holds for supercritical flow. To extend this guarantee to subcritical and transcritical flows, we established two sufficient conditions \eqref{condition on bottom} and \eqref{condition on bottom modified} that connect the bottom profile variation with the speed of propagation and inlet/outlet data, via the Riemann characteristics.  We further proved a Lipschitz stability result, demonstrating that the bottom profile reconstruction depends continuously on the given measurement data, thus quantifying the method's robustness.

In Subsection \ref{testn},  five numerical tests have been detailed for several bottom profiles and inlet discharges, confirming both the nondegeneracy assumption \eqref{condition} and its sufficient conditions \eqref{condition on bottom}-\eqref{condition on bottom modified}.  In all these numerical experiments, the relative $L^{\infty} $ and $L^{2}$ errors \eqref{norms} were computed to rigorously quantify the accuracy of the method. We found that, in all tests, the bed profile is reconstructed with low relative errors. Even when subject to a free surface with noise or a free surface with a high Lipschitz constant $\bigl\|\partial_x \zeta\bigr\|_{L^{\infty}(I)}$, the method remains robust and the bottom profile is well recovered.

The work methodology exploits the exact continuity equation  \eqref{sh4} to recover the spatial discharge profile from the first time derivative of the free surface and the inlet discharge \eqref{invalgo2}. Since this equation follows directly from the full Euler equations without relying on any asymptotic regime, it applies beyond the classical shallow‑water approximation and can be adapted to deeper‑water models that enforce exact mass conservation \cite{coulaud2025comparison}.  Therefore, the methodology has the potential to extend to deeper‑water regimes. In particular, the analytic solution in the steady‑state scenario \eqref{sh5}-\eqref{simpliciation for discharge} also holds for the Peregrine model \cite{peregrine1967long}, because under the assumption of time independence ($\partial_t=0$), this system is equivalent to the shallow water system. This demonstrates that our analytic reconstruction \eqref{sh5}-\eqref{simpliciation for discharge} remains valid in intermediate‑depth.

A promising direction for future work is to extend the work approach to open‐channel flows that admit no stagnation points. Specifically, by relaxing the nondegeneracy assumption  \eqref{condition} to a no‑stagnation requirement, that is,  for every spatial location $x\in I$  there exists a time $t(x)\in (0,t_f)$ such that $\bigl|q(t(x),x)\bigr|>0.$ Developing a reconstruction algorithm based on the inverse system \eqref{sh4} and proving its convergence and stability is of interest.  However, it should be mentioned that this direction may require different time intervals of measurement of the free surface data $S_m$ \eqref{surafcemeasurement}.

\section*{Acknowledgements}
The authors thank Professor Lionel Rosier (ULCO) for his helpful discussion and valuable remarks, which improved this work. The work received France 2030 funding under the reference ‘ANR-23-EXMA-0007'.
\section*{Appendix}
\appendix
\section{Proof of stability}\label{sec:A}
Here, we provide the details for the proof of Proposition \ref{propositionstability}. We recall that, for simplicity,  we omit in the following the explicit dependence on $t^{\star}$.\\
Let $t^{\star}\in (0,t_f)$, such that \eqref{condition} holds. From \eqref{exactsystem} and \eqref{errorsystem}, we have 

\begin{equation}\label{st1}
\begin{cases}
\displaystyle \partial_x \phi=-g \frac{\partial_x \zeta\, q^2}{\phi}\,-\,\partial_t q,\quad \text{in}\;I,\\[6pt]
\displaystyle \partial_x \widetilde{\phi}=-g \frac{\partial_x \widetilde{\zeta}\, (\widetilde{q})^{2}}{\widetilde{\phi}}\,-\,\partial_t\widetilde{ q},\quad \text{in}\; I.
\end{cases}
\end{equation}
Subtracting the equations of the above system \eqref{st1}, we derive 

\begin{equation*}
\begin{aligned}
\partial_x(\phi-\widetilde{\phi}\,)&=-g\,\partial_x\zeta\, q^2\left(\frac{1}{\phi}-\frac{1}{\widetilde{\phi}}\right)+g\,\frac{1}{\widetilde{\phi}}\left(\partial_x\widetilde{\zeta}\,(\widetilde{q})^2-\partial_x\zeta\, q^2\right)-\left(\partial_t q-\partial_t\widetilde{ q}\right)\\
&=-g\,\partial_x\zeta\frac{h\,\widetilde{h}}{\widetilde{q}^2}(\widetilde{\phi}-\phi)+g\,\frac{\widetilde{h}}{\widetilde{q}^2}\biggl(\partial_x\widetilde{\zeta}(\widetilde{q}^2-q^2)-q^2(\partial_x\zeta-\partial_x\widetilde{\zeta}\,)\biggr)-\left(\partial_t q-\partial_t\widetilde{ q}\right).
\end{aligned}
\end{equation*}
Hence,
\begin{equation}\label{st2}
\begin{aligned}
\left|\partial_x(\phi-\widetilde{\phi}\,)\right|&\leq  g \left|\partial_x \zeta\right|\,\frac{h\,\widetilde{h}}{\beta^2}\left|\phi-\widetilde{\phi}\right|+g\frac{\widetilde{h}}{\beta^2}\biggl( \left|\partial_x\widetilde{\zeta}\right| \left| q+\widetilde{q}\right|\left| q-\widetilde{q}\right|+q^2\left| \partial_x\zeta -\partial_x\widetilde{\zeta}\right|  \biggr)   +\left|\partial_t q-\partial_t\widetilde{ q} \right|.               
\end{aligned}
\end{equation}
For ease of reading, we adopt the following notations
$$\bigl\lVert \cdot \bigl\lVert_{L^1(I)}=\bigl\lVert \cdot \bigl\lVert_{1}\quad\text{and}\quad \bigl\lVert \cdot \bigl\lVert_{L^{\infty}(I)}=\bigl\lVert \cdot \bigl\lVert_{\infty}. $$
Employing this notation and integrating the above inequality \eqref{st2} over $(a_1,x)$ for any $x\in I$, we get

\begin{equation}\label{st3}
\begin{aligned}
\displaystyle\int_{a_1}^x  \bigl| \partial_x(\phi -\widetilde{\phi}\,) \bigl| \,{\dd}y &\leq \frac{g}{\beta^2}\bigl\|\partial_x\zeta\bigl\|_{\infty}  \bigl\|h\,\widetilde{h}\bigl\|_{\infty}  \displaystyle\int_{a_1}^x\bigl|\phi-\widetilde{\phi}\bigl|\,{\dd}y  \, \\&\quad\quad+\frac{g}{\beta^2}\bigl\|\widetilde{h}\bigl\|_{\infty}\biggl( \bigl\|\partial_x\widetilde{\zeta}(q+\widetilde{q}\,)\bigl\|_{\infty}\,\bigl\|q-\widetilde{q}\bigl\|_1+\bigl\|q^2\bigl\|_{\infty}\bigl\|\partial_x(\zeta-\widetilde{\zeta}\,)\bigl\|_1\biggr)\\
&\quad\quad\quad+\bigl\|\partial_t q-\partial_t\widetilde{ q}\bigl\|_1.
\end{aligned}
\end{equation}
By substituting the constants in  \eqref{stabilitycanst}, the preceding inequality \eqref{st3} is equivalent to

\begin{equation*}
\begin{aligned}
\displaystyle\int_{a_1}^x  \bigl| \partial_x(\phi -\widetilde{\phi}\,) \bigl| \,{\dd}y &\leq C_2  \displaystyle\int_{a_1}^x\bigl|\phi-\widetilde{\phi}\bigl|\,{\dd}y+C_3\bigl\|q-\widetilde{q}\bigl\|_1+C_4\bigl\|\partial_x(\zeta-\widetilde{\zeta}\,)\bigl\|_1+\bigl\|\partial_t q-\partial_t\widetilde{ q}\bigl\|_1.
\end{aligned}
\end{equation*}
On the other hand, since $\phi$ and $\widetilde{\phi}$  are in $W^{1,1}(I)$, we have
$$\bigl|\phi(x)-\widetilde{\phi}(x)\bigr|\leq \bigl|\phi(a_1)-\widetilde{\phi}(a_1)\bigr|+\displaystyle\int_{a_1}^x  \bigl| \partial_x(\phi -\widetilde{\phi}\,) \bigr| \,{\dd}y,\quad \forall x\in I.$$
Then, for any $x\in I$, we get
$$\bigl|\phi(x)-\widetilde{\phi}(x)\bigr|\leq \bigl|\phi(a_1)-\widetilde{\phi}(a_1)\bigr|+C_3\bigl\|q-\widetilde{q}\bigl\|_1+C_4\bigl\|\partial_x(\zeta-\widetilde{\zeta}\,)\bigl\|_1+\bigl\|\partial_t q-\partial_t\widetilde{ q}\bigl\|_1+C_2  \displaystyle\int_{a_1}^x\bigl|\phi-\widetilde{\phi}\bigl|\,{\dd}y.$$
With algebraic manipulation, we can obtain the following
\begin{equation*}
\begin{aligned}
\bigl|\phi(a_1)-\widetilde{\phi}(a_1)\bigr|&=\left|\frac{q^2(a_1)}{h(a_1)}-\frac{\widetilde{q}^2(a_1)}{\widetilde{h}(a_1)}\right|\\
&=\frac{1}{h(a_1)\widetilde{h}(a_1)}\bigl|\widetilde{h}(a_1)q^2(a_1)-h(a_1)\widetilde{q}^2(a_1)\bigr|\\
&=\frac{1}{h(a_1)\widetilde{h}(a_1)}\left|q^2(a_1)\left(\widetilde{h}(a_1)-h(a_1)\right)+h(a_1)\left(q^2(a_1)-\widetilde{q}^2(a_1)\right)\right|\\
&\leq \frac{q^2(a_1)}{h(a_1)\widetilde{h}(a_1)}\Bigl( \bigl|b(a_1)-\widetilde{b}(a_1)\bigr|+\bigl|\zeta(a_1)-\widetilde{\zeta}(a_1)\bigr| \Bigr)+\frac{q(a_1)+\widetilde{q}(a_1)}{\widetilde{h}(a_1)}\bigl|q(a_1)-\widetilde{q}(a_1)\bigr|.
\end{aligned}
\end{equation*}
Hence, the definition of $E$ in \eqref{E} yields
$$\bigl|\phi(x)-\widetilde{\phi}(x)\bigr|\leq E+C_2 \displaystyle\int_{a_1}^x\bigl|\phi-\widetilde{\phi}\bigl|\,{\dd}y,\quad \forall x\in I.$$
Applying Grönwall’s inequality \cite{evans2022partial}, we obtain
\begin{equation}\label{Granwall}
\displaystyle \big\|\phi -\widetilde{\phi}\bigr\|_1\leq \frac{E}{C_2}\Bigl(e^{C_2(a_2-a_1)}-1\Bigr).
\end{equation}
We now use the definitions of $\phi$ and $\widetilde{\phi}$  to transform the above stability \eqref{Granwall} into a stability estimate in terms of the bathymetry profile $b$. We first link the bottom profile error to the error on $\phi$ as follows:
\begin{equation*}
\begin{aligned}
\bigl| b-\widetilde{b}\bigr|&=\bigl|(h-\widetilde{h})-(\zeta-\widetilde{\zeta})\bigr|\\
&\leq \bigl|h-\widetilde{h}\bigr|+ \bigl|\zeta-\widetilde{\zeta}\bigr|\\
&=\frac{1}{\phi\,\widetilde{\phi}}\bigl|q^2 \widetilde{\phi}-\widetilde{q}^2 \phi\bigr|+\bigl|\zeta-\widetilde{\zeta}\bigr|\\
&\leq \frac{h\,\widetilde{h}}{\widetilde{q}^2}\bigl|\widetilde{\phi}-\phi\bigr|+\frac{\widetilde{h}}{\widetilde{q}^2}\bigl|q^2-\widetilde{q}^2\bigr|+\bigl|\zeta-\widetilde{\zeta}\bigr|.
\end{aligned}
\end{equation*}
By integration over $I$, we get
\begin{equation}\label{st4}
    \bigl\|b-\widetilde{b}\bigr\|_1\leq \frac{\bigl\|h\,\widetilde{h}\bigr\|_{\infty}}{\beta^2}\bigl\|\phi-\widetilde{\phi}\bigr\|_1+\frac{1}{\beta^2}\bigl\|\widetilde{h}(q+\widetilde{q})\bigr\|_{\infty}\bigl\|q-\widetilde{q}\bigr\|_1+\bigl\|\zeta-\widetilde{\zeta}\bigr\|_1.
\end{equation}
Injecting \eqref{Granwall} into the above inequality \eqref{st4}  and using the definition of $C_1$ in \eqref{stabilitycanst}, results in \eqref{stability}, thus completing the proof.



\begin{thebibliography}{10}


\bibitem{de1871theorie}
A. B.~De Saint-Venant.
\newblock{ \em Théorie du mouvement non permanent des eaux, avec application aux crues des rivières et à lintroduction des marées dans leur lit.}
\newblock {C. R. Acad. Sc.}, 73:147154, 1871


\bibitem{lannes2013water}
D.~Lannes.
\newblock {\em The water waves problem: mathematical analysis and asymptotics}, volume~188.
\newblock Amer. Math. Soc., Providence RI, 2013

\bibitem{forbes1988critical}
L. K.~Forbes.
\newblock {\em Critical free-surface flow over a semi-circular obstruction.}
\newblock {J. Eng. Math}, 22:13--13, 1988

\bibitem{fadda1997open}
D.~Forbes, P.~Raad.
\newblock {\em Open channel flow over submerged obstructions: an experimental and numerical study.}
\newblock { J. Fluid. Eng.-T. ASME}, 119:906, 1997

\bibitem{tam2015predicting}
A.~Tam, Z.~Yu, R. M.~Kelso, B. J.~Binder.
\newblock {\em Predicting channel bed topography in hydraulic falls.}
\newblock { Phys. Fluids }, 27:112106, 2015

\bibitem{angel2024bathymetry}
J.~Angel, J.~Behrens, S.~G$\ddot{o}$tschel, M.~Hollm, D.~Ruprecht, and R.~Seifried.
\newblock {\em Bathymetry reconstruction from experimental data using PDE-constrained optimisation.}
\newblock {Computers \& Fluids}, 278:106321, 2024


\bibitem{huang2018detecting}
C.~Huang, Y.~Chen, S.~Zhang, J.~Wu.
\newblock {\em Detecting, extracting, and monitoring surface water from space
using optical sensors: a review.}
\newblock {Rev. Geophys.}, 56:333–360, 2018

\bibitem{lee2006electromagnetic}
J-S.~Lee, P.Y.~Julien.
\newblock {\em Electromagnetic wave surface velocimetry.}
\newblock { J. Hydraul. Eng.}, 132:146--153, 2006

\bibitem{adrian2011particle}
R. J.~Adrian, R. J.~Westerweel.
\newblock {\em  Particle image velocimetry}.
\newblock UK: Cambridge Univ. Press., Cambridge, 2011


\bibitem{gessese2011reconstruction}
A. F.~Gessese, M.~Sellier, E.~Van Houten, G.~Smart.
\newblock {\em Reconstruction of river bed topography from free surface data using a direct numerical approach in one-dimensional shallow water flow.}
\newblock { Inverse Problems}, 27:025001, 2011


\bibitem{gessese2012direct}
A. F.~Gessese, M.~Sellier.
\newblock {\em A direct solution approach to the inverse shallow-water problem}
\newblock { Math. Probl. Eng.}, 2012:417950, 2012


\bibitem{hajduk2020bathymetry}
H.~Hajduk, D.~Kuzmin,  V.~Aizinger.
\newblock {\em Bathymetry reconstruction using inverse Shallow Water models: Finite element discretization and regularization}.
\newblock { In: Lecture notes in computational science and engineering}, volume~132,  Springer, Cham., 2020



\bibitem{gessese2011inferring}
A. F.~Gessese, M.~Sellier, E.~Van Houten, G.~Smart.
\newblock {\em Inferring channel bed topography from known free surface data.}
\newblock { 34th IAHR World Congress}, Brisbane, Australia, 2011


\bibitem{petcu2013one}
M.~Petcu, R.~Temam.
\newblock {\em The one-dimensional shallow water equations with transparent boundary condition.}
\newblock { Math. Methods Appl. Sci.}, 36:1979--1994, 2013


\bibitem{huang2011one}
A.~Huang, M.~Petcu, R.~Temam.
\newblock {\em The one-dimensional supercritical shallow-water equations with topography.}
\newblock { Ann. Univ. Buchar. Math. Ser.}, 2:63--82, 2011

\bibitem{kounadis2020galerkin}
G.~Kounadis, V. A.~Dougalis.
\newblock {\em Galerkin finite element methods for the Shallow Water equations over variable bottom.}
\newblock { J. Comput. Appl. Math.}, 373:112315, 2020


\bibitem{khan2021variational}
R.A.~Khan, N.K.-R.~Kevlahan.
\newblock {\em Variational assimilation of surface wave data for bathymetry reconstruction. Part I: algorithm and test cases.}
\newblock { Tellus A: Dyn. Meteorol. Oceanogr.}, 73:1--25, 2021

\bibitem{sanders2000adjoint}
R. F.~Sanders, N. D.~Katopodes.
\newblock {\em Adjoint sensitivity analysis for shallow-water wave control.}
\newblock { J. Eng. Mech.}, 126:909--919, 2000

\bibitem{vasan2013inverse}
V.~Vasan, B.~Deconinck.
\newblock {\em The inverse water wave problem of bathymetry detection.}
\newblock { J. Fluid Mech.}, 714:562--590, 2013

\bibitem{iguchi2021hyperbolic}
T.~Iguchi, D.~Lannes.
\newblock {\em Hyperbolic free boundary problems and applications to wave-structure interactions.}
\newblock { Indiana. Univ. Math. J.}, 70:353--4640, 2021

\bibitem{evans2022partial}
L. C.~Evans.
\newblock {\em Partial Differential Equations (in Appendix B)}.
\newblock { Amer. Math. Soc.}, volume~19, Providence, RI, 2010


\bibitem{kurganov2007second}
A.~Kurganov, G.~Petrova.
\newblock {\em A second-order well-balanced positivity preserving central-upwind scheme for the Saint-Venant system.}
\newblock { Commun. Math. Sci.}, 5:133–160, 2007

\bibitem{kurganov2002central}
A.~Kurganov, D.~Levy.
\newblock {\em Central-upwind schemes for the Saint-Venant system.}
\newblock { Math. Mod. and Num. Anal.}, 36:397--425, 2002


\bibitem{codepyhon}
N.~Lamsahel.
\newblock {\em Python code for the direct approach of bottom profile detection.}
\newblock {https://github.com/Lamsahel-Noureddine/Python-code-for-the-direct-approach-of-bottom-profile-detection }, 2025

\bibitem{gottlieb2001strong}
S.~Gottlieb, C.-W.~Shu, E.~Tadmor.
\newblock {\em Strong stability-preserving high-order time discretization methods.}
\newblock { SIAM Rev.}, 43:89--112, 2001


\bibitem{butcher2000numerical}
J. C.~Butcher.
\newblock {\em Numerical methods for ordinary differential equations in the 20th century.}
\newblock {J. Comput. Appl. Math.}, 125:1--292, 2000


\bibitem{liang2009adaptive}
Q.~Liang, A. G. L.~Borthwick.
\newblock {\em Adaptive quadtree simulation of shallow flows with wet--dry fronts over complex topography.}
\newblock { Comput. Fluids}, 38:221--234, 2009

\bibitem{nicholls2009detection}
D. P.~Nicholls, M.~Taber.
\newblock {\em Detection of ocean bathymetry from surface wave measurements.}
\newblock { Eur. J. Mech. B/Fluids}, 28:224--233, 2009


\bibitem{reinsch1967smoothing}
C. H.~Reinsch.
\newblock {\em Smoothing by spline functions.}
\newblock { Numer. Math.}, 10:177--183, 1967



\bibitem{coulaud2025comparison}
G.~Coulaud, M.~Teles, M.~Benoit.
\newblock {\em A comparison of eight weakly dispersive Boussinesq-type models for non-breaking long-wave propagation in variable water depth.}
\newblock { Coast. Eng.}, 195:104645, 2025


\bibitem{peregrine1967long}
D. H.~Peregrine.
\newblock {\em Long waves on a beach.}
\newblock { J. Fluid Mech.}, 27:815--827, 1967



\bibitem{grilli1994shoaling}
S. T.~Grilli, R.~Subramanya, I. A.~Svendsen, J.~Veeramony.
\newblock {\em Shoaling of solitary waves on plane beaches.}
\newblock { J. Waterways Port, Coast. Ocean Eng.}, 120:609--628, 1994

\bibitem{marks2000integration}
K.~Marks, P.~Bates.
\newblock {\em Integration of high-resolution topographic data with floodplain flow models.}
\newblock {Hydrol. process.}, 14:2109--2122, 2000


\bibitem{synolakis1991green}
C. E.~ Synolakis
\newblock {\em Green’s law and the evolution of solitary waves.}
\newblock { Phys. Fluids A Fluid Dyn.}, 3:490, 1991

\bibitem{synolakis1987runup}
C. E.~ Synolakis
\newblock {\em The run-up of solitary waves.}
\newblock {J. Fluid Mech.}, 185:523--545, 1987



\bibitem{cunge1980practical}
J. A.~Cunge, F. M.~Holly, A.~Verwey.
\newblock {\em Practical aspects of computational river hydraulics.}
\newblock {Pitam Publishing}, London, 1980


\bibitem{synolakis2008validation}
C. E.~Synolakis, E. N.~Bernard, V. V.~Titov, U.~K\^{a}no\u{g}lu, F. I.~González.
\newblock {\em Validation and verification of tsunami numerical models.}
\newblock {Pure Appl. Geophys.}, 165:2197--2228, 2008



\bibitem{smith2004conventional}
W. H. F.~Smith, D. T.~Sandwell.
\newblock {\em Conventional bathymetry, bathymetry from space, and geodetic altimetry.}
\newblock {Oceanography}, 17:8--23, 2008


\bibitem{carron2001proposed}
M. J.~Carron, P. R.~Vogt, W.-Y. Jung.
\newblock {\em A proposed international long-term project to systematically map the world's ocean floors from beach to trench: GOMaP (Global Ocean Mapping Program).}
\newblock {Int. Hydrogr. Rev.}, 2:49-–55, 2001


\bibitem{lecours2016review}
V.~Lecours, M. F. J.~Dolan, A~.Micallef, V. L.~Lucieer.
\newblock {\em A review of marine geomorphometry, the quantitative study of the seafloor.}
\newblock {Hydrol. Earth Syst. Sci.}, 20:3207--3244, 2016



\bibitem{irish1999scanning}
J. L.~ Irish, W. J.~Lillycrop.
\newblock {\em Scanning laser mapping of the coastal zone: the SHOALS system.}
\newblock {ISPRS J. of Photogramm.  remote sens.}, 54:123--129, 1999



\bibitem{hilldale2008assessing}
R. C.~Hilldale, D.~Raff.
\newblock {\em Assessing the ability of airborne LiDAR to map river bathymetry.}
\newblock {Earth Surf. Process. Landf.}, 33:773--783, 2008


\bibitem{sagawa2019satellite}
T.~Sagawa, Y.~Yamashita, T.~Okumura, T.~Yamanokuchi.
\newblock {\em Satellite derived bathymetry using machine learning and multi-temporal satellite images.}
\newblock {Remote Sens.}, 11:1155, 2019


\bibitem{smart2009river}
G. M.~Smart, J.~Bind, M. J.~Duncan.
\newblock {\em River bathymetry from conventional LiDAR using water surface returns.}
\newblock {18th World IMACS / MODSIM Congress}, 2009



\bibitem{sellier2016inverse}
M.~Sellier.
\newblock {\em Inverse problems in free surface flows: a review.}
\newblock {Acta Mech.}, 227:913--935, 2016


\bibitem{khan2022variational}
R. A.~Khan, N. K.-R.~Kevlahan.
\newblock {\em Variational assimilation of surface wave data for bathymetry reconstruction. Part II: Second order adjoint sensitivity analysis.}
\newblock {Tellus A: Dyn. Meteorol. and Oceanogr.}, 74:187–203, 2022






\bibitem{gessese2013bathymetry}
A. F.~Gessese, K. M.~Wa, M.~Sellier.
\newblock {\em Bathymetry reconstruction based on the zero-inertia shallow water approximation.}
\newblock {Theor. Comput. Fluid Dyn.}, 27:721--732, 2013


\bibitem{cocquet2021optimization}
P.-H.~Cocquet, S.~Riffo, J.~Salomon.
\newblock {\em Optimization of bathymetry for long waves with small amplitude.}
\newblock {SIAM J. Control Optim.}, 59:4429--4456, 2021

%
%


\end{thebibliography}
\end{document}